\documentclass[conference]{IEEEtran}

\IEEEoverridecommandlockouts

\usepackage[breaklinks=true]{hyperref}
\hypersetup{colorlinks,citecolor=blue,allcolors=blue}

\usepackage{enumitem} 
\usepackage{amsfonts}
\usepackage[english]{babel}
\usepackage[utf8]{inputenc}
\usepackage{graphicx}
\usepackage{amssymb,amsmath,amsthm}
\usepackage{mathrsfs}
\usepackage{mathtools}
\usepackage{bm}
\usepackage{adjustbox}   
\usepackage{subfig}
\usepackage{multirow}

\usepackage{tabularx,colortbl}
\usepackage{colortbl}
\usepackage{graphicx}
\usepackage{listings}
\usepackage{cleveref}
\usepackage{soul}
\usepackage{url}
\usepackage{setspace}
\usepackage{makecell}
\usepackage[ruled,noend]{algorithm2e}
 \usepackage{algpseudocode}

\def\BibTeX{{\rm B\kern-.05em{\sc i\kern-.025em b}\kern-.08emT\kern-.1667em\lower.7ex\hbox{E}\kern-.125emX}}

\date{\today}
\newtheorem{thm}{Theorem}[section]

\newtheorem{defn}[thm]{Definition}

\newcommand{\R}{\mathbb{R}}

\newcommand{\tsr}[1]{\pmb{\mathcal{#1}}}

\newcommand{\vcr}[1]{\mathbf{#1}}
\newcommand{\mat}[1]{\mathbf{#1}}

\newcommand{\inti}[2]{\{{#1},\ldots, {#2}\}}

\newcommand{\name}[1]{{\color{blue}[name] }}

\usepackage[normalem]{ulem}

\usepackage{pythonhighlight}

\usepackage{listings}
\definecolor{mygreen}{rgb}{0,0.2,0}
\definecolor{mygray}{rgb}{0.5,0.5,0.5}
\definecolor{mymauve}{rgb}{0.58,0,0.82}
\definecolor{mypurple}{rgb}{0.38,0,0.32}
\definecolor{myblue}{rgb}{0.1,0,0.32}
\newcommand{\costyle}{\footnotesize\ttfamily\bfseries}
\newcommand{\kwstyle}{\costyle\textcolor{myblue}}
\begin{document}
\pagestyle{plain}

\title{AutoHOOT: Automatic High-Order Optimization for
Tensors}

    \author{\IEEEauthorblockN{Linjian Ma$^{*}$\thanks{$^{*}$Equal contribution.}}
    \IEEEauthorblockA{\textit{Department of Computer Science} \\
    \textit{University of Illinois at Urbana-Champaign}\\
    Urbana, IL \\
    lma16@illinois.edu}
    \and
    \IEEEauthorblockN{Jiayu Ye$^{*}$}
    \IEEEauthorblockA{\textit{Google}\\
    Sunnyvale, CA \\
    yejiayu@google.com}
    \and
    \IEEEauthorblockN{Edgar Solomonik}
    \IEEEauthorblockA{\textit{Department of Computer Science} \\
    \textit{University of Illinois at Urbana-Champaign}\\
    Urbana, IL \\
    solomon2@illinois.edu}
    }

\maketitle

\begin{abstract}
High-order optimization methods, including Newton's method and its variants as well as alternating minimization methods, dominate the optimization algorithms for tensor decompositions and tensor networks. These tensor methods are used for data analysis and simulation of quantum systems.
In this work, we introduce AutoHOOT, the first automatic differentiation (AD) framework targeting at high-order optimization for tensor computations. AutoHOOT takes input tensor computation expressions and generates optimized derivative expressions.
In particular, AutoHOOT contains a new explicit Jacobian / Hessian expression generation kernel whose outputs maintain the input tensors' granularity and are easy to optimize. The expressions are then optimized by both the traditional compiler optimization techniques and specific tensor algebra transformations. Experimental results show that AutoHOOT achieves competitive CPU and GPU performance for both tensor decomposition and tensor network applications compared to existing AD software and other tensor computation libraries with manually written kernels. The tensor methods generated by AutoHOOT are also well-parallelizable, and we demonstrate good scalability on a distributed memory supercomputer.
\end{abstract}

\begin{IEEEkeywords}
automatic differentiation, computational graph optimization, tensor computation, tensor decomposition, tensor network
\end{IEEEkeywords}

\section{Introduction}
\label{sec:intro}

Tensors, represented as multidimensional arrays in the computer program, are important in both scientific computing and machine learning. Tensor decomposition~\cite{kolda2009tensor} is a powerful tool in compressing and approximating the high dimensional data, and is used widely in numerical PDEs~\cite{pazner2018approximate}, quantum chemistry~\cite{hohenstein2012tensor,hummel2017low} and statistical modeling~\cite{anandkumar2014tensor,sidiropoulos2017tensor}.
Tensor networks are also widely used in  physics to approximate quantum states~\cite{markov2008simulating,orus2014practical} and in neural networks to form tensorized neural architectures~\cite{novikov2015tensorizing}.
Convolution, which is a basic tensor operation, is widely used in computer vision applications~\cite{krizhevsky2012imagenet}.
Tensors are also widely used in methods for electronic structure calculations in computational chemistry~\cite{hirata2003tensor}.

Derivatives, mostly in the form of gradients, are ubiquitous in the optimization algorithms for tensor related problems. For neural networks, they are used to calculate the gradients of the loss function w.r.t. the model parameters. For tensor decomposition and tensor networks, first-order and higher-order derivatives are necessary to construct the operators used in the alternating optimization. 
Gradients of computational chemistry methods are used for optimization of the electronic geometry to identify stable states and state transitions~\cite{jorgensen2012geometrical}.
Automatic differentiation (AD) frameworks, including popular Python tools such as PyTorch~\cite{paszke2019pytorch}, JAX~\cite{jax2018github}, and TensorFlow~\cite{abadi2016tensorflow}, can generate derivatives in all of these contexts.
However, in tensor decomposition, tensor networks, and quantum chemistry, gradient calculations are most often done via manually written codes, as careful numerical and performance considerations are required in these more complex settings.

AD transforms a software or mathematical expression of a function into code for computation of its derivatives with respect to the desired parameters.
Although mathematically correct, the output programs for the derivatives may be sub-optimal in computational cost, use of efficient kernels such as the BLAS, memory footprint, and numerical stability.
Components of different frameworks address these problems jointly or independently. For example, transformations of the computational graph and operator fusion are used to improve computational efficiency and parallelizability~\cite{abadi2016tensorflow,jia2019taso,paszke2019pytorch}.
Gradient checkpointing and garbage collection are used to address memory bottlenecks~\cite{abadi2016tensorflow,paszke2019pytorch}.
For large scale tensor computations, computational and memory demands leave little leeway for error in these aspects.

Common commercial AD frameworks such as PyTorch~\cite{paszke2019pytorch}, JAX~\cite{jax2018github}, and TensorFlow~\cite{abadi2016tensorflow} are focused on  first-order numerical optimization methods on deep learning models. 
In the context of tensor decompositions, tensor network optimization, and differentiation of tensor methods, three major additional challenges arise.
\begin{enumerate}[topsep=0pt,leftmargin=*]
    \item These domains predominantly employ alternating second-order optimization methods, as they provide monotonic convergence and rapid progress at almost the same per-iteration cost as first-order methods. These methods employ implicit representations of the Jacobian and Hessian to solve linear systems. Existing AD frameworks have limited logical constructs for second-order derivative information, and consequently generate code that can be sub-optimal in cost by orders of magnitude.
    \item Most tensor operations involved in the deep learning applications are related to small tensors, while in tensor network and tensor decomposition applications, there are many tensor contractions over high order (multidimensional) tensors with a large number of elements. Therefore, tensor network applications require better optimization algorithms to select optimized contraction order and eliminate redundant calculations.
    \item Deep learning computational graphs usually have large depth with many nonlinear operations, making the freedom to optimize tensor operations limited. On the other hand, in tensor decomposition and tensor network applications, the computational graphs are usually wide and have small depth, so there is more freedom to optimize the computation. 
\end{enumerate}
Although many frameworks, such as Tensorly~\cite{kossaifi2019tensorly}, TensorNetwork~\cite{roberts2019tensornetwork} and Quimb~\cite{gray2018quimb}, provide interfaces to optimize the tensor decomposition / networks algorithms with AD frameworks such as TensorFlow and PyTorch, the optimization algorithms are the general first-order methods and its variants. These frameworks explicitly implement popular second-order methods for these problems, such as Alternating Least Squares (ALS) for tensor decompositions and Density Matrix Renormalization Group (DMRG) for 1D tensor networks, rather than using AD.
 The ability to generate efficient expressions of these methods automatically via AD, would accelerate the development of new variants and their deployment on shared-memory, GPU, and distributed-memory architectures.

In this paper, we propose a new AD framework for tensor computations, Automatic High-Order Optimization for Tensors (AutoHOOT). The library is publicly available at \url{https://github.com/LinjianMa/AutoHOOT}. AutoHOOT encapsulates the following novel ideas and capabilities:
\begin{itemize}[topsep=0pt,leftmargin=*]
    \item
    a new AD module that generates more efficient representations for higher-order derivative constructs such as Jacobians and Hessians, which are needed for tensor computation applications,
    \item a new computational graph optimization module that extends beyond the traditional optimization techniques for compilers with tensor-algebra specific transformations, such as distributivity of matrix inversion over the Kronecker product,
    \item portability via high-level support for different tensor contraction backends: NumPy for multi-core CPU, TensorFlow for GPUs, and Cyclops~\cite{solomonik2014massively} for distributed memory systems,
    \item substantial improvements in sequential and parallel performance for tensor network and tensor decomposition optimizations over other AD libraries and competitive or improved performance w.r.t. manually-optimized implementations.
\end{itemize}

\section{Background}
\label{sec:background}
\subsection{Notations and Definitions}

For vectors, bold lowercase Roman letters are used, e.g., $\vcr{x}$. For matrices, bold uppercase Roman letters are used, e.g., $\mat{X}$. For tensors, bold calligraphic uppercase Roman letters are used, e.g., $\tsr{X}$. An order $N$ tensor corresponds to an $N$-dimensional array with dimensions $s_1\times \cdots \times s_N$. 

Elements of vectors, matrices, and tensors are denoted in parentheses, e.g., $\vcr{x}(i)$ denotes the $i$th entry of a vector $\vcr{x}$, $\mat{X}(i,j)$ denotes the $(i,j)$th element of a matrix $\mat{X}$, and $\tsr{X}(i,j,k,l)$ denotes the $(i,j,k,l)$th element of an order 4 tensor $\tsr{X}$. Subscripts are used to label different vectors, matrices,  tensors and functions (e.g. $\tsr{X}_1$ and $\tsr{X}_2$, $f_1$ and $f_2$). 

Matricization is the process of unfolding a tensor into a matrix. Given a tensor $\tsr{X}$ the mode-$n$ matricized version is denoted by $\mat{X}_{(n)}\in \R^{s_n\times K}$ where $K=\prod_{m=1,m\neq n}^N s_m$. We generalize this matricization definition, so that $\mat{X}_{(i:j)}$ means that the dimensions from the $i$th index to the $j$th index are unfolded to the column dimension of the matrix, and all the other dimensions are unfolded to the row dimension of the matrix.

For a scalar output function $y=f(\vcr{a}_1, \ldots, \vcr{a}_N)$, We use the $\vcr{g}^{[f]}_{[\vcr{a}_i]}$ and $\mat{H}^{[f]}_{[\vcr{a}_i]}$ to denote the gradient vector and Hessian matrix of $f$ w.r.t the input vectors $\vcr{a}_i$. When the inputs are tensors, the gradient and the Hessian will also be a tensor and denote $\tsr{G}^{[f]}_{[\tsr{A}_i]}$ and $\tsr{H}^{[f]}_{[\tsr{A}_i]}$. For a function with non-scalar output $\vcr{y} = f(\vcr{a}_1, \ldots, \vcr{a}_N)$, we use $\mat{J}^{[f]}_{[\vcr{a}_i]}$ to denote the Jacobian matrix of the function $f$ w.r.t one of the input vectors $\vcr{a}_i$. The shape of the Jacobian matrix will be $\R^{|
\vcr{y}|\times|\vcr{a}_i|}$. If $\tsr{Y}$ is an output tensor with size $\R^{s_1\times \ldots \times s_M}$, and $\tsr{A}_i$ is an input tensor with size $\R^{r_1\times \ldots\times r_K}$, then the Jacobian will be a tensor denoted as $\tsr{J}^{[f]}_{[\tsr{A}_i]}$ with dimensions $\R^{s_1\times \ldots\times s_M\times r_1\times\ldots\times r_K}$. 

We also define generalized Vector Jacobian Product (VJP), Jacobian Vector Product (JVP) and Hessian Vector Product (HVP). 
When both Jacobian and Hessian are matrices, these are matrix-vector multiplication operations. 
When Jacobian and Hessian are both tensors defined above, these are tensor contractions, whose results are the same as unfolding the tensors into matrices and performing the matrix-vector product.

\subsection{Numerical Optimization Algorithms for Tensor Computations}
\label{subsec:numericalopt}

We consider two tensor numerical problems: the nonlinear least squares fitting and the eigenvalue problem. 
For both problems, we denote $\tsr{X}$ as the input tensor which can be an explicit tensor or implicit tensor network (e.g., Matrix Product Operator~\cite{verstraete2004matrix}), $f$ as a tensor network function and $\tsr{A}_1, \ldots, \tsr{A}_N$ as the optimization variables. Then the objective for the nonlinear least squares problem is defined as
\begin{equation}
\min_{\tsr{A}_1, \ldots, \tsr{A}_N}\phi(\tsr{A}_1, \ldots, \tsr{A}_N) := \frac{1}{2}\|\tsr{X}- f(\tsr{A}_1, \ldots, \tsr{A}_N)\|^2,
\label{eq:leastsquare}
\end{equation}
which finds a generalized low rank approximation of the input tensor $\tsr{X}$. The objective for the eigenvalue problem is defined as
\begin{equation}
\min_{\tsr{A}_1, \ldots, \tsr{A}_N}\psi(\tsr{A}_1, \ldots, \tsr{A}_N) := 
\frac{
\vcr{v}^T_{(1:N)} \mat{X}_{(1:N)}\vcr{v}_{(1:N)}
}
{
\|\tsr{V}\|_F^2
}
,
\label{eq:eigenvalue}
\end{equation}
where $\tsr{V} = f(\tsr{A}_1, \ldots, \tsr{A}_N)$
and the output of $f$ serves as a generalized low rank approximation of the eigenvector of a Hermitian matrix that is a matricization of $\tsr{X}$.

Three categories of algorithms are generally used to optimize the problems: second-order methods, including Newton's method and its variants, alternating minimization, which updates each input / site at one time, and first-order methods such as gradient descent and its variants. 

\textbf{Newton's method and its variants.}
Newton's method and its variants, such as Gauss-Newton (GN) method, are popular methods to solve non-linear least squares problems for a quadratic objective function defined in Equation~\ref{eq:leastsquare}. Let $\vcr{a}$ denote the concatenation of all the vectorized sites $\text{vec}(\tsr{A}_i)$
and $\hat{f}(\vcr a)=\text{vec}(f(\tsr{A}_1,\ldots,\tsr{A}_N))$, so that
$
r(\vcr{a}) := \text{vec}(\tsr{X})- \hat{f}(\vcr{a})
$
denotes the vectorized residual.
Further, let $r_i(\vcr{a})$ denote the $i$th element of the output of function $r$.
The gradient and the Hessian matrix of $\phi$ can be expressed as
\begin{align*}
\nabla \phi (\vcr{a}) = & \mat{J}^{[r]T}_{[\vcr{a}]}r(\vcr{a}),  \\
\mat{H}^{[\phi]}_{[\vcr{a}]} = \mat{J}^{[r]T}_{[\vcr{a}]}\mat{J}^{[r]}_{[\vcr{a}]}& + \sum_{i}r_i(\vcr{a})\mat{H}^{[r_i]}_{[\vcr{a}]}.
\end{align*}
The Newton iteration performs the update based on
\[\vcr{a}^{(k+1)} = \vcr{a}^{(k)}- (\mat{H}^{[\phi]}_{[\vcr{a}^{(k)}]})^{-1} \mat{J}_{[\vcr{a}^{(k)}]}^{[r]T}r(\vcr{a}^{(k)}),\]
while the Gauss-Newton method leverages the fact that $\mat{H}^{[r_i]}_{[\vcr{a}]}$ is negligible as its norm is small when the residual is small, therefore the update can be performed as
\[\vcr{a}^{(k+1)} = \vcr{a}^{(k)}- (\mat{J}^{[r]T}_{[\vcr{a}^{(k)}]}\mat{J}^{[r]}_{[\vcr{a}^{(k)}]})^{-1} \mat{J}_{[\vcr{a}^{(k)}]}^{[r]T}r(\vcr{a}^{(k)}),\]
where $\vcr{a}^{(k)}$ represents the $\vcr{a}$ at $k$th iteration.
The Gauss-Newton updates can be regarded as normal equations for the linear least squares problem. Both Newton and Gauss-Newton methods can be solved via the conjugate gradient method
with matrix-vector products performed with an implicit form of the Jacobian / Hessian
to avoid costly matrix inversion~\cite{tichavsky2013further,singh2019comparison}.

\textbf{Alternating minimization.} For tensor numerical optimization, in many cases both the input and output dimensions are large, and it's computationally expensive to form the explicit Hessian / Jacobian matrix w.r.t. all the variables and perform the second-order method directly. 
On the other hand, when optimizing a subset of variables, forming the Hessian or Jacobian with respect to those variables is affordable and effective.
Most often, alternating minimization procedures update one tensor operand at a time.
For Equation~\ref{eq:leastsquare}, such subproblem can be formulated as
\begin{equation}
  \min_{\tsr{A}_i}\phi(\tsr{A}_1, \ldots, \tsr{A}_N).
  \label{eq:als}
\end{equation}
Each $\tsr{A}_i$ for $i\in \{1,\ldots, N\}$ is updated once via its subproblem during an optimization sweep. For tensor decompositions and tensor networks, each subproblem is often quadratic, allowing for the minima to be found directly, often at a similar cost to updating $\tsr{A}_i$ with a first-order method.
Alternating minimization also generally provides monotonic convergence.

In each sweep, many terms necessary to form the subproblems have many equivalent intermediates, and choosing the proper contraction paths to form and also amortize them can greatly save the cost. 
This scheme, called
the \textit{dimension tree} algorithm, is critical to the algorithm performance, and has been used in both tensor decompositions~\cite{phan2013fast,vannieuwenhoven2015computing,ma2018accelerating} and DMRG to save the cost.

\textbf{First-order methods.} 
The efficacy of the first-order methods on tensor computations is dependent on the applications.
The first-order methods are shown to be advantageous on achieving high fitting accuracies on some tensor decomposition problems~\cite{acar2011scalable}, while they also perform worse 
than alternating minimization in achieving high accuracy for large scale tensor completion problems~\cite{zhang2019enabling}. 
The per-iteration cost of first-order methods is often comparable to that of both second-order methods and the alternating minimization method, due to the structure of tensor networks $f$ in Equations~\ref{eq:leastsquare}, \ref{eq:eigenvalue}.

Traditional AD frameworks can generate efficient kernels for first-order methods, while their performance on the kernels in higher-order methods is sub-optimal. 
In this paper, we focus on the performance optimization over both second-order method and alternating minimization methods, to accelerate future development of efficient high-order methods for various applications.
However, we believe our graph optimization techniques also have the potential to produce efficient formulations for first-order methods, where the objective involves the contractions of high-order tensors, which arise in quantum chemistry methods~\cite{hirata2003tensor}.

\subsection{Previous Work}
Optimization for tensor computations requires three essential building blocks, automatic differentiation, optimization of the generated set of tensor operations, and a computational backend for individual tensor operations.
Existing software for tensor computations, including Tensorly~\cite{kossaifi2019tensorly}, TensorNetwork~\cite{roberts2019tensornetwork} and Quimb~\cite{gray2018quimb}, permit the use of multiple backends for individual tensor operations, and provide some constructs to make use of AD. However, when using AD, these libraries employ general AD backends such as JAX or TensorFlow in a black-box fashion.

Automatic differentiation is generally provided via one of two ways, operator overloading~\cite{paszke2019pytorch,jax2018github,tokui2019chainer,walter2013algorithmic,maclaurin2015autograd} or source code transformation (SCT)~\cite{abadi2016tensorflow,van2018tangent,DBLP:journals/corr/JiaSDKLGGD14}. Operator overloading requires the user to write functions in terms of the provided library constructs and constructs the derivatives at run-time,
while SCT uses precompilation to generate code for derivative computation.
Operator overloading provides a similar mental programming model as normal computer programs~\cite{tokui2019chainer}, yielding code that is easier to interpret and debug than SCT.
On the flip side, SCT has more potential to optimize the computational graph with global graph information. Consequently, SCT is generally the method of choice for AD libraries that aim to achieve high performance (e.g.,~\cite{abadi2016tensorflow}).

Our work on graph optimization builds on substantial efforts for optimization of computational graphs of tensor operations.
Tensor contraction can be optimized via parallelization~\cite{rajbhandari2014communication,kendall2000high,Kats_sp_tensor2013,solomonik2014massively}, efficient transposition~\cite{hptt2017}, blocking~\cite{choi2018blocking,li2018hicoo,ibrahim2014analysis,senanayake2019unified}, exploiting symmetry~\cite{hirata2003tensor,solomonik2014massively,solomonikfast}, and sparsity~\cite{kjolstad2017tensor,peng2016massively,Kats_sp_tensor2013,manzer2017general,peng2016massively,10.1145/2833179.2833183}.
For complicated tensor graphs, specialized compilers like XLA~\cite{xla2017xla} and TVM~\cite{TVM2018} rewrite the computational graph to optimize program execution and memory allocation on dedicated hardware. For machine independent optimization, Grappler in TensorFlow~\cite{abadi2016tensorflow} and TASO~\cite{jia2019taso} use rule based symbolic substitution to simplify the execution flow.
Classical compiler optimization also includes relevant techniques such as common subexpression elimination \cite{aho1986compilers} are widely used as well \cite{abadi2016tensorflow,hartono2006identifying}. 
Previous work, such as Opt\_einsum~\cite{smith2018opt} has yielded approaches for automatically determining efficient contraction orderings and selecting the best intermediates~\cite{hirata2003tensor,auer2006automatic,hartono2009performance,hartono2005automated}.
The approaches generally rely on heuristic or exhaustive search to select a contraction path, as finding the optimal contraction order is NP-hard \cite{chi1997optimizing}.

\section{Overall Architecture}
\label{sec:arch}

\begin{figure}[]
\centering
\includegraphics[width=1.3in]{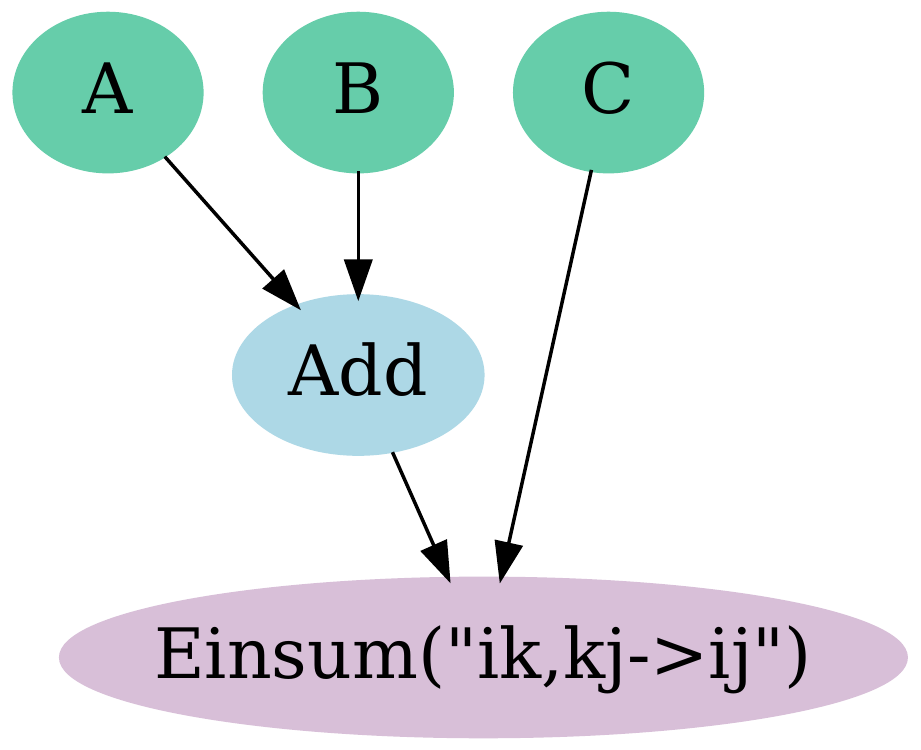}
\caption{An example of a computational graph. We use green nodes to denote input variables, purple nodes to denote output nodes, and blue nodes to denote intermediate or constant nodes.}
\label{fig:graph_example}
\end{figure}

\begin{figure}[]
\centering
\includegraphics[width=.48\textwidth]{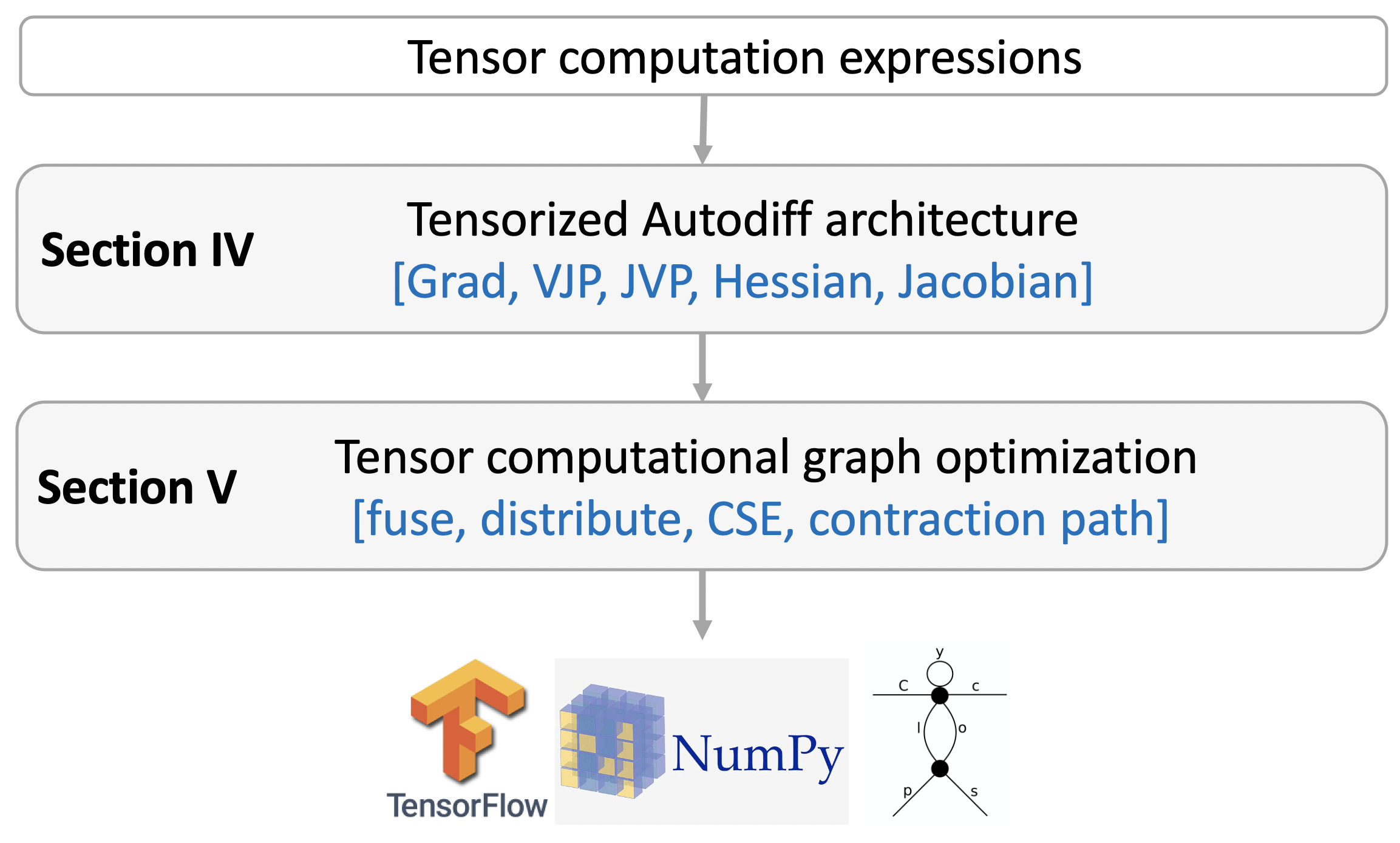}
\caption{System overview of AutoHOOT. The arrows show the computation flow.}
\label{fig:overview}
\end{figure}

The computations in AutoHOOT
are described by \textit{computational graphs}, which are directed graphs revealing the data dependency between different operations. Each \textit{node} can be a source, intermediate or sink. Source / Sink nodes are inputs / outputs of the graph. Sink and intermediate nodes can be any mathematical computation, while input nodes are fed by the user or constants.  An \textit{edge} connecting two nodes represents the data dependency between them. An example of a computational graph is shown in Figure~\ref{fig:graph_example}, where $\mat{A}, \mat{B}, \mat{C}$ are source nodes, the Einsum node is the sink, and the graph computes $(\mat{A}+\mat{B})\mat{C}$. We typically refer a node with its type, e.g., an Einsum node, which represents the tensor computations based on the Einstein summation convention.
An \textit{Einsum graph} is defined as a graph of nodes where all the nodes except the sources are Einsum nodes. 
An \textit{Einsum tree} is defined as a tree of nodes where all the nodes except the sources are Einsum nodes.

AutoHOOT has two major components: an automatic differentiation architecture for tensor computations and a tensor computational graph optimizer. Figure~\ref{fig:overview} shows the system overview. For an input computation expression, the AD module will generate its tensorized differentiation expressions. Both the input expressions and the differentiation expressions will be optimized through the graph optimization module. 
With the optimized expressions, users have the choice to directly run the optimized expressions using the framework backends, including NumPy, TensorFlow and Cyclops, or to generate the Python source code through the source generation module.

Below we show an example to perform the CP decomposition based on alternating least squares using the framework. Rather than constructing each subproblem and building the dimension tree based algorithm manually, we only need to construct the updates of Newton's method for each subproblem, and the \textit{optimize} function will reorganize the computational graph to minimize execution time automatically. 
\begin{python}
# construct input expressions
A, B, C, input_tensor, loss = cpd_graph(size, rank)

def update_site(site):
    hes = ad.hessian(loss, [site])
    grad, = ad.gradients(loss, [site])
    new_site = ad.tensordot(
        ad.tensorinv(hes[0][0]), grad)
    # return the optimized computational graph
    return optimize(new_site)
new_A = update_site(A)
new_B = update_site(B)
new_C = update_site(C)

# This executor is shared among all updates. 
executor = ad.Executor([loss, new_A, new_B, new_C])
# ALS iterations
for i in range(num_iter):
    A_val = executor.run(feed_dict={
        input_tensor: input_tensor_val,
        A: A_val, B: B_val, C: C_val
    }, out=[new_A])
    B_val = executor.run(feed_dict={
        input_tensor: input_tensor_val,
        A: A_val, B: B_val, C: C_val
    }, out=[new_B])
    C_val = executor.run(feed_dict={
        input_tensor: input_tensor_val,
        A: A_val, B: B_val, C: C_val
    }, out=[new_C])
    loss_val = executor.run(feed_dict={
        input_tensor: input_tensor_val,
        A: A_val, B: B_val, C: C_val
    }, out=[loss])
\end{python}

In the AD module, we implement the reverse mode AD for first-order derivatives (Jacobian, VJP and JVP), as well as for higher-order derivatives, including Hessian and HVP. Both Jacobian and Hessian are formulated with a new algorithm, such that their calculations are not dependent on the JVP and HVP routines, which is more amenable to
 parallel execution as well as graph optimizations. We describe this approach in detail in Section~\ref{sec:diff}.

The graph optimizer provides optimizations for tensor computational graphs.
We adopt many machine independent optimization algorithms for common tensor computational graphs, such as selection of optimal contraction path and common sub-expression elimination. For second-order methods, the graph optimizer rewrites the structured inverse, such as the inverse of a Kronecker product, so that the inverses are operated on smaller tensors. For alternating methods, we developed a path selection algorithm with constraints to construct the dimension trees. We describe this algorithm
in detail in Section~\ref{sec:opt}.

\section{Computational Graphs for High-Order Derivatives}
\label{sec:diff}

We implement the reverse-mode AD based on the source code transformation (SCT) method, explicitly transforming the primal computation expression prior to execution to the adjoint expression. It allows us to flexibly perform the computational graph optimization after the adjoint expression production.

Our AD module supports the operations which calculate the Jacobian / Hessian expressions implicitly (VJP, JVP and HVP), and also explicit Jacobian and Hessian calculations. The implicit calculations are widely used in many other frameworks, because it is computationally cheaper. For example, for a Hessian matrix with size $n \times n$, explicitly forming the matrix costs $O(n^2)$, while the HVP calculation will only cost $O(n)$ leveraging the back-propagation gradient functions. For the explicit Jacobian and Hessian calculations, we introduce a new back-propagation algorithm that can produce a computational graph is more amenable to parallelization and downstream optimizations. The algorithm is detailed in Section~\ref{subsec:explicithessian}.

\subsection{VJP, JVP, and HVP}

Our implementation of VJP is similar to many other frameworks~\cite{paszke2019pytorch,jax2018github,abadi2016tensorflow}, and is based on the reverse-mode AD. For functions involving matrix / vector operations whose inputs and outputs are both vectors,
\[
\vcr{x}_{i+1} = f_i(\vcr{x}_i), i\in [1,\ldots, N],
\]
consider a computational graph consisting of a chain of these functions,
\[
\vcr{y} = f(\vcr{x}_1) = f_N\cdots f_1(\vcr{x}_1),
\]
the VJP adjoint of $\vcr{x}_i$, $\vcr{v}^T\mat{J}^{[f]}_{[\vcr{x}_i]}$, is calculated based on the VJP adjoint of $\vcr{x}_{i+1}$,
{\small
\[
\text{VJP}(\vcr{v}, f, \vcr{x}_i) {=} \vcr{v}^T\mat{J}^{[f]}_{[\vcr{x}_i]} {=} (\vcr{v}^T\mat{J}^{[f]}_{[\vcr{x}_{i+1}]})\mat{J}^{[f_{i}]}_{[\vcr{x}_{i}]} {=} \text{VJP}(\vcr{v}, f,  \vcr{x}_{i+1})\mat{J}^{[f_{i}]}_{[\vcr{x}_{i}]}.
\]}
Therefore, the VJP of all the inputs / intermediates $\vcr{x}_i$, $i\in [1,\ldots, N]$ will be calculated with one backward propagation. It is also computationally efficient, because only matrix-vector product is necessary for each calculation.

Note that for the cases where sub function inputs and outputs contain matrices or tensors, VJP with reverse-mode AD is still valid and efficient, since we can think of each  matrix or tensor as a reshaped vector. For the case where the output is a scalar, the gradient expression is implemented based on the VJP, if we fix the vector as a unit length vector with element being one.

 Our JVP implementation is based on the VJP function\footnote{The JVP implementation is based on the technique introduced at \url{https://j-towns.github.io/2017/06/12/A-new-trick.html}.}.
 Although it's more computationally efficient to implement JVP based on forward mode AD~\cite{baydin2017automatic}, we choose to implement it based on our reverse mode AD module, and optimize the computational graph afterwards to achieve computationally efficient expressions. 
 The JVP implementation is based on calling the VJP function twice. First, we construct a function $g$, whose expression is as follows,
\[
    g(\vcr{u}) = \text{VJP}(\vcr{u}, f, \vcr{x})^T =  (\vcr{u}^T \mat{J}^{[f]}_{[\vcr{x}]})^T.
\]
Afterwards, we perform another VJP operation on the function $g$ with related to its input $\vcr{u}$, and can get the JVP expression,
\small
\begin{align*}
    \text{VJP}(\vcr{v}, g, \vcr{u})^T = 
    (\vcr{v}^T \mat{J}^{[g]}_{[\vcr{u}]} )^T 
    = (\vcr{v}^T \mat{J}^{[f]T}_{[\vcr{x}]})^T = \mat{J}^{[f]}_{[\vcr{x}]}\vcr{v} 
    = \text{JVP}(\vcr{v}, f, \vcr{x}) .
\end{align*}

\normalsize
We also implement the HVP function based on the gradient function. We only consider the case when the function output is a scalar, because it is the general case where Hessian matrices are used. 
The HVP is formulated based on two gradient calculations, because HVP is equivalent to the gradient of the gradient-vector inner product. The expression is shown as follows,
\begin{align*}
\text{HVP}(\vcr{v}, f, \vcr{x}) = 
 \mat{H}^{[f]}_{[\vcr{x}]}\vcr{v} = 
 \frac{\partial{\vcr{g}^{[f]}_{[\vcr{x}]}}}{\partial{\vcr{x}}} \vcr{v} =
\frac{\partial{\vcr{g}^{[f]}_{[\vcr{x}]}}}{\partial{\vcr{x}}}\vcr{v} + \vcr{g}^{[f]T}_{[\vcr{x}]}\frac{\partial{\vcr{v}}}{\partial{\vcr{x}}} \\
=  \frac{\partial{(\vcr{g}^{[f]T}_{[\vcr{x}]}\vcr{v})}}{\partial{\vcr{x}}}
= \text{grad}(\text{grad}(f, \vcr{x})^T \vcr{v}, \vcr{x} ).
\end{align*}

\normalsize
\subsection{Explicit Jacobian and Hessian}
\label{subsec:explicithessian}

To the best of our knowledge, all of the popular AD frameworks calculate explicit Jacobian and Hessian based on the VJP and HVP routines~\cite{abadi2016tensorflow,jax2018github,paszke2019pytorch}. Taking the Jacobian calculation of 
\begin{align*}
    f(\vcr{x})= \mat{A}_1\mat{A}_2\vcr{x}
\end{align*}
as an example: when both $\vcr{x}$ and $f(\vcr{x})$ are of size $n$, current methods will compute the $i$th row of the Jacobian via VJP $\vcr{e}_i^T\mat{J}^{[f]}_{[\vcr{x}]}$ for $i\in \{1,\ldots,n\}$, where $\vcr{e}_i$ is the $i$th elementary vector.
There are two major disadvantages to this approach:
\begin{itemize}[topsep=0pt,leftmargin=*]
    \item It changes the BLAS-3 level matrix-matrix multiplications to multiple BLAS-2 level matrix-vector multiplications, and less flop intensity can be achieved. Although many frameworks provide the routine to compute all the matrix-vector multiplications in parallel, the parallelism is still sub-optimal and less efficient than the matrix multiplications, because the flop-to-byte ratio is $O(1)$ versus $O(n)$.
    \item The computational graph produced is difficult to optimize. Although having high dimensions, many Jacobians / Hessians in tensor computation operations are highly structured and the computational cost can be greatly reduced if being well optimized. However, calculating them based on matrix-vector products adds one more matrix-vector product operation, which usually break the structure and increase the cost. 
    
    For example, if $\mat{A}_1 = \mat{B}\otimes \mat{C}$ and $\mat{A}_2 = \mat{D} \otimes \mat{E}$ and $\mat{B}, \mat{C}, \mat{D}, \mat{E}$ have sizes $n\times n$, performing matrix-vector product for the Jacobian and each elementary vector costs $O(n^4)$ and the overall Jacobian calculation cost is $O(n^6)$. However, if we calculate the Jacobian directly, we can use the mixed-product property of the Kronecker product to optimize the expression,
    \[
    (\mat{B}\otimes \mat{C})(\mat{D}\otimes \mat{E}) = (\mat{BD}) \otimes (\mat{CE}),
    \]
    reducing the overall cost to $O(n^4)$.
\end{itemize}
To alleviate these disadvantages, we produce both Jacobian and Hessian expressions in a way that's independent of VJP and HVP routines. 

For the Jacobian expression, our implementations are also based on the chain rule to perform back propagation, using
\[
\text{Jacobian}(f, \vcr{x}_i) {=} \mat{J}^{[f]}_{[\vcr{x}_i]} {=} \mat{J}^{[f]}_{[\vcr{x}_{i+1}]}\mat{J}^{[f_i]}_{[\vcr{x}_{i}]} {=} \text{Jacobian}(f, \vcr{x}_{i+1})\mat{J}^{[f_i]}_{[\vcr{x}_{i}]}.
\]
Therefore, the Jacobian of one target node is the matrix-matrix product between the Jacobian of its output node and the Jacobian of the local function. Note that when both $\vcr{x}_i$ and the Jacobian have the tensor format, the above equation still holds, except that the matrix-matrix product is expressed in the form of tensor contractions (Einsums).

For linear operations, such as addition, subtraction, scalar-tensor multiplication and Einsum, we formulate the Jacobian expressions as an Einsum. To achieve that, we introduce the \textit{Identity node}, which is a node that applies an identity matrix, to express the constraints in Jacobian tensors. For example, for the addition operations of two order $N$ tensors,
\[
    f(\tsr{A}, \tsr{B}) = \tsr{A} + \tsr{B}, 
\]
its Jacobian is a tensor of order $2N$,
where $\tsr{J}^{[f]}_{[\tsr{A}]}(x_1, \ldots, x_{2N}) = 1$ if and only if $x_i=x_{i+N}$ for $i\in\{1,\ldots, N\}$, and other elements are 0. This constraint can be easily specified with identity nodes. For the order 3 addition, the Jacobian of $\tsr{A}$ can be expressed as 
    \[
    \tsr{J}^{[f]}_{[\tsr{A}]}(i,j,k,l,m,n)
    = \mat{I}(i,l)\mat{I}(j,m)\mat{I}(k,n).
    \]
Similarly, we can use the method to express the Jacobians for all the other linear operations. For example, for an Einsum expression below, its Jacobians are written as
    \[
    f(\tsr{A}, \tsr{B})(i,j,k)= \sum_{l}\tsr{A}(i,k,l)\tsr{B}(j,k,l),
    \]
    \[
    \tsr{J}^{[f]}_{[\tsr{A}]}(i,j,k,m,n,o)=
    \mat{I}(i,m)\mat{I}(k,n)\tsr{B}(j,n,o),
    \]
    \[
    \tsr{J}^{[f]}_{[\tsr{B}]}(i,j,k,m,n,o)=
    \mat{I}(j,m)\mat{I}(k,n)\tsr{A}(i,n,o).
    \]
Although we have introduced several identity nodes, they can be easily pruned so that only necessary identity nodes are left, which will be introduced in Section~\ref{sec:opt}. The Hessian routines are based on the Jacobian routines: we perform Jacobian calculations twice to get the Hessian expressions. The advantage of this Jacobian / Hessian generation method is three-fold:
first, we can leverage BLAS-3 level operations to perform most of the tensor contractions and can achieve higher performance. Second, the expressions are much easier to optimize, as will be introduced below. Third, the source code for Jacobian / Hessian expressions can be easily acquired, which is beneficial for both debugging and research purposes.

\section{Graph Optimizations}
\label{sec:opt}
\begin{figure*}
\centering
\subfloat[Einsum distribution.]{
\includegraphics[width=.495\textwidth]{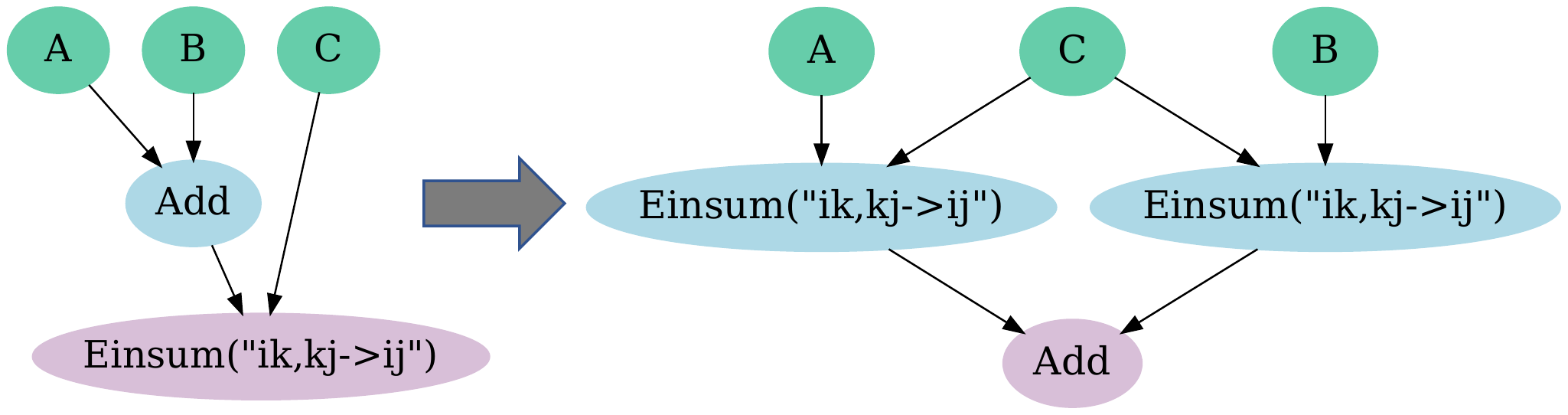}
\label{subfig:distribute}
}
\subfloat[Identity node pruning. The nodes whose name starts from "I" are identity nodes.]{
\includegraphics[width=.495\textwidth]{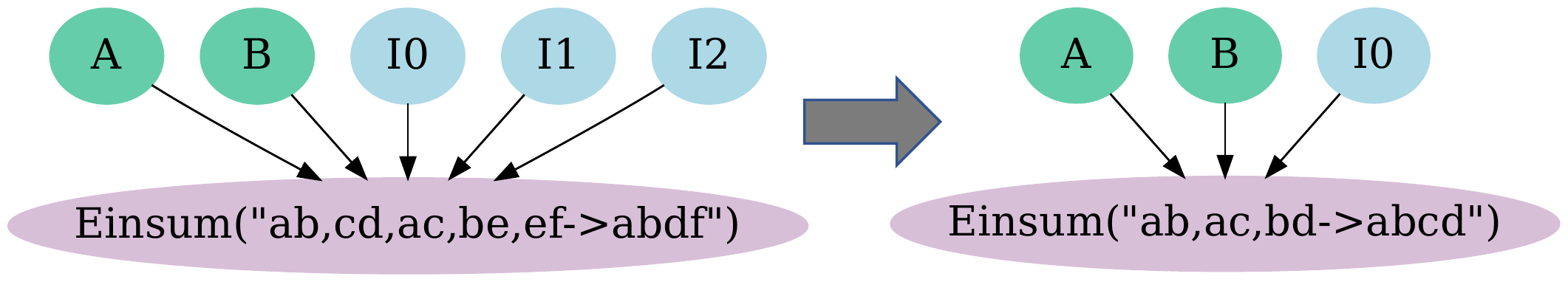}
\label{subfig:identity}
}

\subfloat[Einsum fusion. It transforms an Einsum graph into one single Einsum node.]{
\includegraphics[width=.99\textwidth]{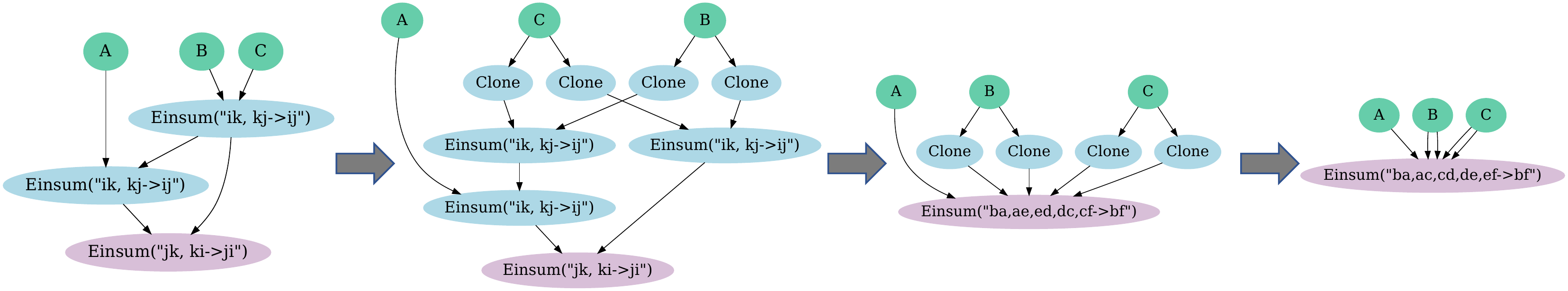}
\label{subfig:fusion}
}

\subfloat[Optimization of tensor inversion]{
\includegraphics[width=.68\textwidth]{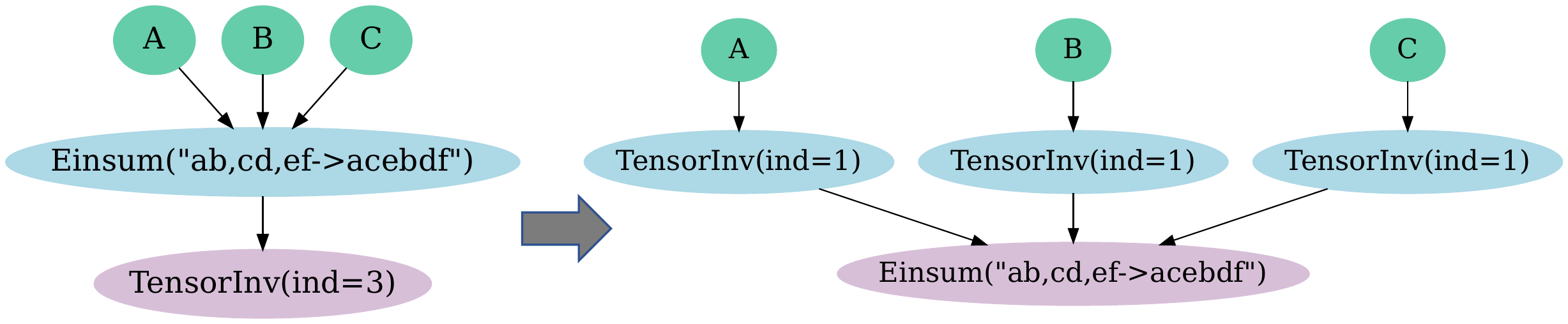}
\label{subfig:optinverse}
}
\subfloat[Inverse node pruning]{
\includegraphics[width=.3\textwidth]{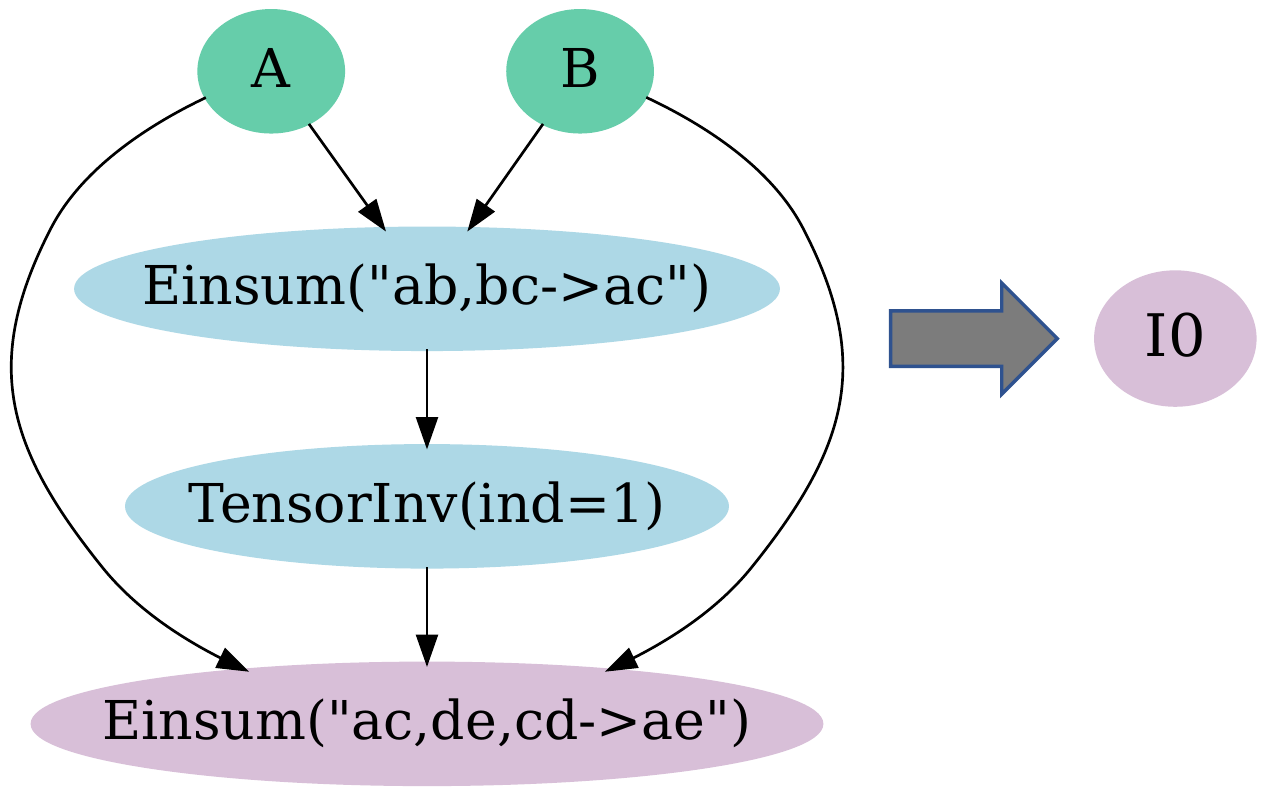}
\label{subfig:pruneinverse}
}

\subfloat[Common subexpression elimination]{
\includegraphics[width=.495\textwidth]{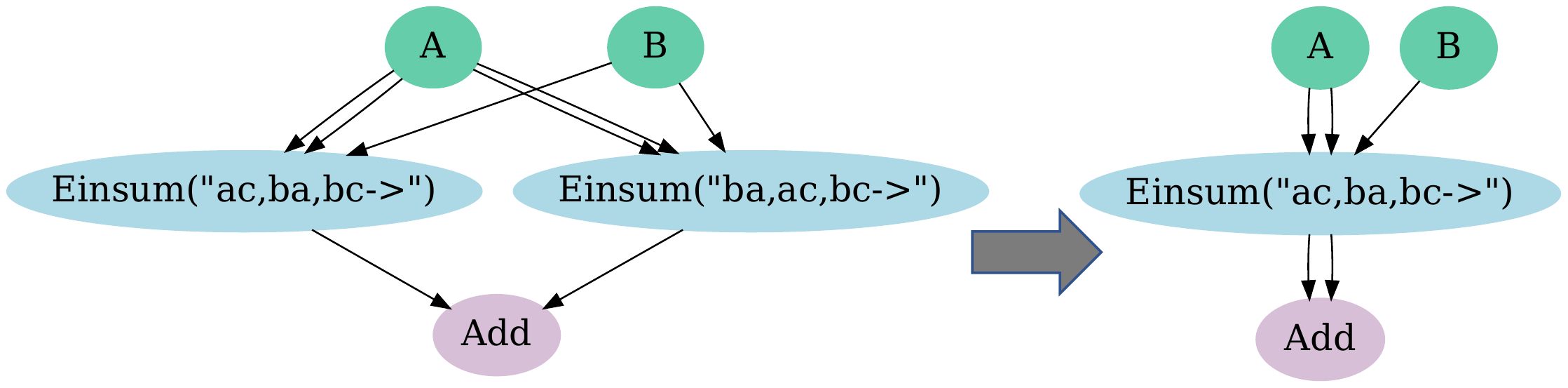}
\label{subfig:cse}
}
\subfloat[Optimal contraction path]{
\includegraphics[width=.495\textwidth]{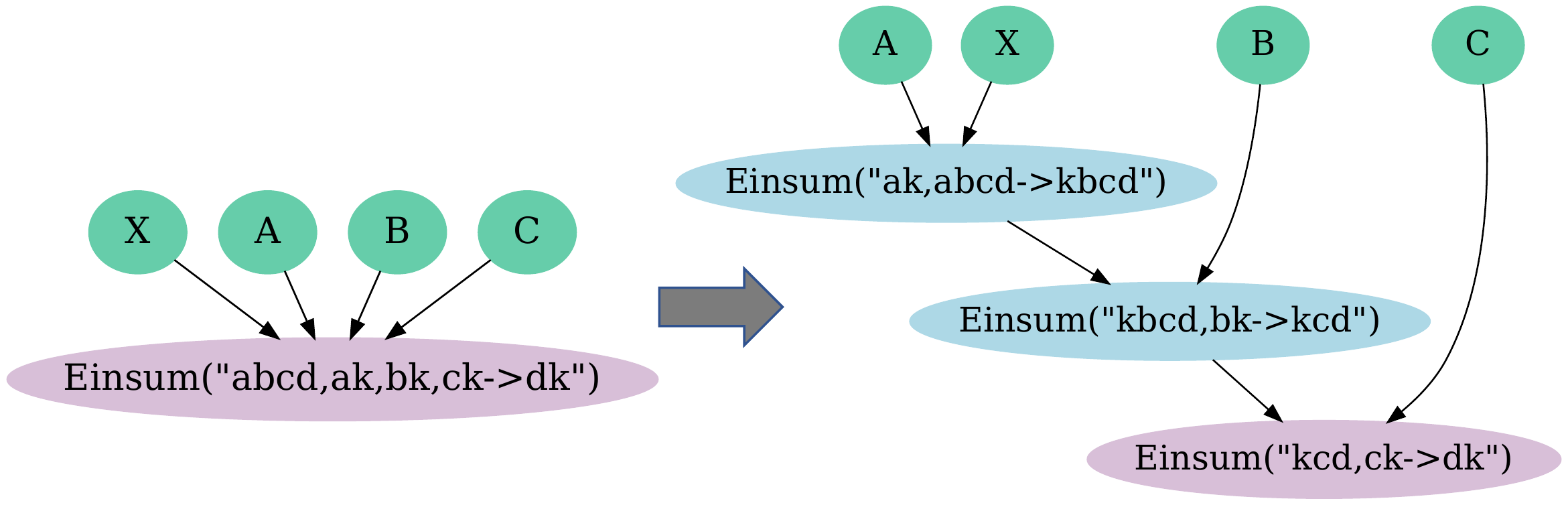}
\label{subfig:optpath}
}

\subfloat[Dimension tree generation]{
\includegraphics[width=.99\textwidth]{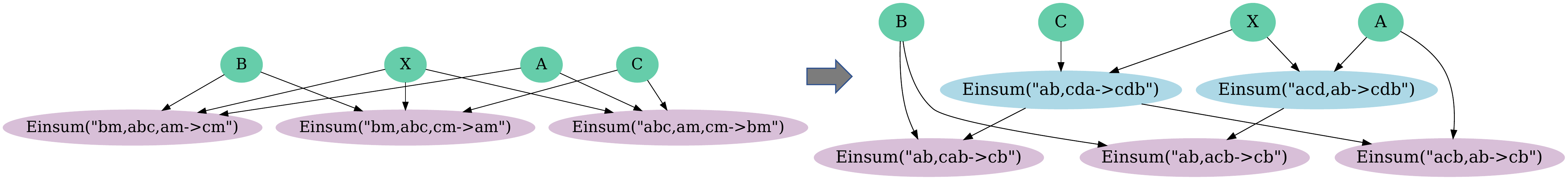}
\label{subfig:dt}
}

\caption{Visualization of different graph optimization kernels.}
\end{figure*}

\begin{figure}[]
\centering
\includegraphics[width=.35\textwidth]{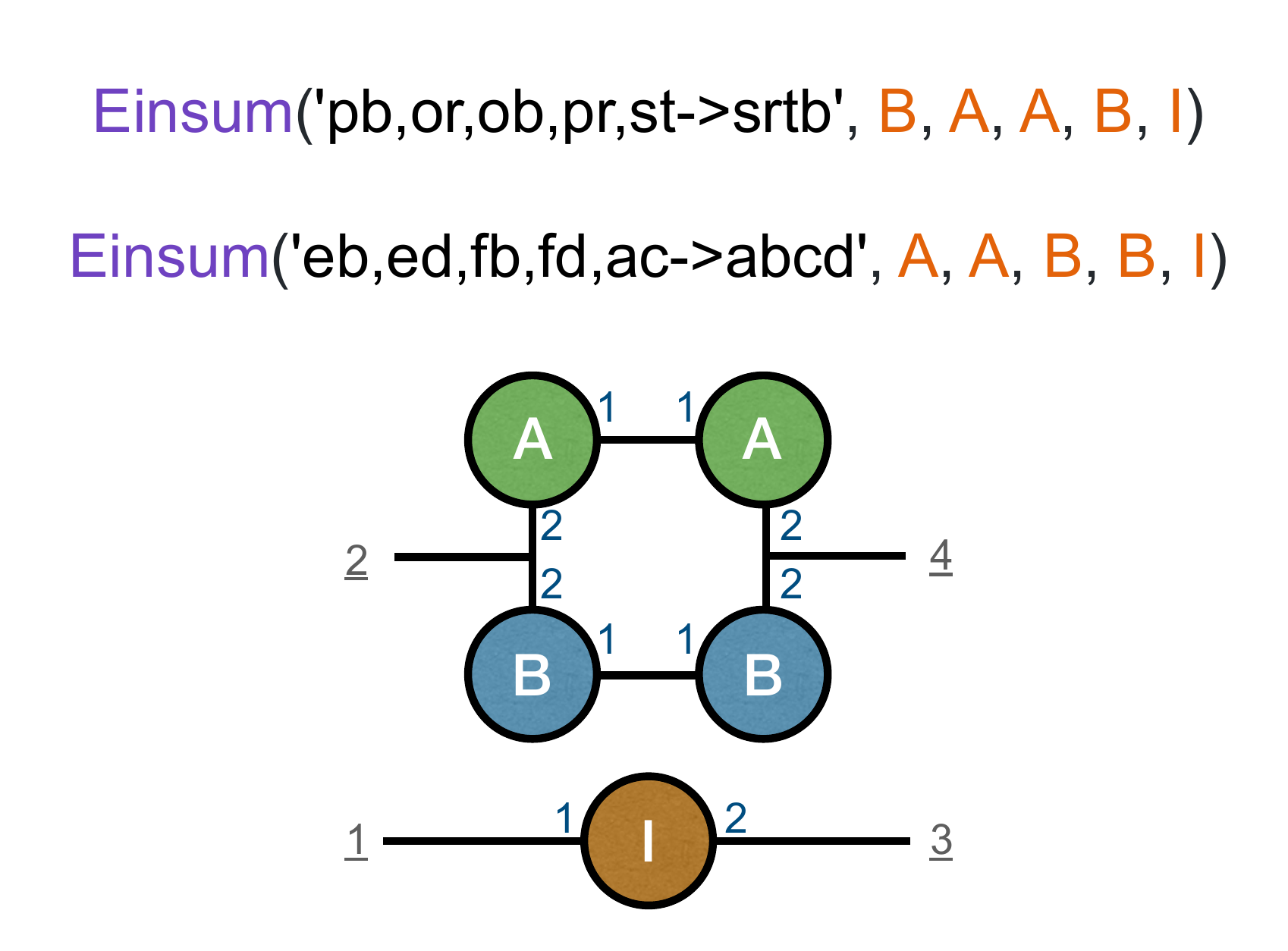}
\caption{Tensor diagram of two Einsum expressions with the same tensor computations. The numbers around the input tensor denote the dimension numbers that are contracted by specific edges. The underlined numbers denote the dimension number of the output tensor. Two Einsum expressions with the same tensor diagram express the same tensor computations.}
\label{fig:csediagram}
\end{figure}

We built a compiler to optimize tensor computational graphs. The compiler is specifically designed for tensor expressions with multilinear operations, including tensor contractions (Einsum) and linear algebra operations (addition, multiplication, summation, inversion and so on).
Our goal is to reduce the computational cost by transforming the graph to an equivalent form. 
Given the fact that retrieving the optimal execution graph is NP-hard, we devise several application-driven heuristic strategies: 
\begin{itemize}[topsep=0pt,leftmargin=*]
    \item {\it Generation of longer Einsum nodes:} To achieve this, we implement two kernels, Einsum distribution and Einsum fusion.  
    \item {\it Symbolic rule execution:} We implement the structured inverse node decomposition and redundant node pruning kernels. In addition, we use SymPy~\cite{meurer2017sympy} to simplify elementary algebraic operations.
    \item {\it Contraction order selection:} We select the contraction path on fully simplified expressions. 
    \item {\it Constrained contraction path construction:} To accelerate alternating minimization, we provide a kernel to reuse intermediates between optimization subproblems.
\end{itemize}
Traditional compiler techniques, such as common sub-expressions elimination, are applied after the strategies above. The overall algorithm is described in Algorithm~\ref{alg:overall}.

\begin{algorithm}[t]
\DontPrintSemicolon
\caption{Graph optimization}
\label{alg:overall}
    \SetAlgoLined
    \SetKwInOut{Input}{input}\SetKwInOut{Output}{output}
    \Input{Input Graph: G}
    \Output{Optimized Graph: OG}
    \SetKwFunction{Distribution}{Distribution}
    \SetKwFunction{FuseAllEinsum}{FuseAllEinsum}
    \SetKwFunction{PruneIdentity}{PruneIdentity}
    \SetKwFunction{OptimizeInverse}{OptimizeInverse}
    \SetKwFunction{OptContractPath}{OptContractPath}
    \SetKwFunction{SympySimplify}{SympySimplify}
    \SetKwFunction{CSE}{CSE}
    \SetKwFunction{SymbolicExecution}{SymbolicExecution}
    
    G = \FuseAllEinsum(\Distribution(G))  \Comment{Provide longer Einsums}
    
    G = \SymbolicExecution(G) \Comment{Decompose Inverse / Prune identity / SymPy}
    
    G = \OptContractPath(G) \Comment{ Find efficient contraction order}
    
    OG = \CSE(G) \Comment{Common Subexpression Elimination}
    
    return OG
\end{algorithm}

\subsection{Longer Einsum Nodes Generation}

We aim to transform the computational graph into Einsum nodes with as many inputs as possible. This optimization will empower the contraction path selection with a global view and ease the discovery of optimizable patterns for downstream algorithms. To achieve this, we introduce two transformation kernels. 

\textbf{Einsum distribution.} 
Einsum distribution recursively leverages distributivity of tensor contraction over tensor addition (or another distributive operation) to generate larger Einsum graphs. Larger Einsum graphs are the prerequisite for further graph depth reduction. This optimization moves the nodes performing the distributive operation (\texttt{dist\_op}) closer to the graph sinks based on the programmatic rule below. Figure~\ref{subfig:distribute} illustrates the idea of an application of the algorithm while the pseudo-code can be found in Algorithm~\ref{alg:distribute} in the Appendix.
\begin{python}
Einsum(dist_op(g1, g2), g3) = 
dist_op(Einsum(g1, g3), Einsum(g2, g3))
\end{python}

\textbf{Einsum fusion.} Einsum fusion transforms an Einsum graph into several distinct Einsum nodes with the same set of source vertices (inputs) leveraging associativity of tensor contractions. It is a prerequisite for downstream graph optimization steps, such as contraction path selection and identity node pruning. An example can be seen in Figure~\ref{subfig:fusion}.

Einsum fusion has three steps: linearization of the graph, fusion of the generated Einsum Tree, and removal of the redundant clone nodes. The linearization step changes the input Einsum graph into an Einsum tree. When a source node is used in multiple Einsums, we create a clone of it for each Einsum. If an Einsum node has more than one output, we copy the subgraph defining its computation, including itself, and repeat until all nodes have a single output, yielding a forest (set of disconnected trees). The fusion step fuses each generated Einsum tree. It leverages a union-find data structure, which puts two dimensions from two Einsum nodes into one set if they have the same subscript in one Einsum expression. After that, each disjoint set is assigned an unique character for the generation of the subscript of the new Einsum node. 
Finally, the clone node removal step removes the redundant clone nodes and returns an Einsum node. We illustrate both the pseudo-code sketch of the algorithm and the union-find data structure in Algorithm~\ref{alg:ufbuilder} and \ref{alg:fusion} in Appendix~\ref{subsec:appendix_opt}.

\subsection{Symbolic Execution}
We employ several linear algebra constructs that can simplify the computational graph and reduce the computational cost.

\textbf{Structured Tensor inverse decomposition.} An inverse of an Einsum graph may be the bottleneck of the computational graph because of the cubic order complexity. Fortunately, structured information may guide the optimization, e.g, the inverse of a Kronecker Product can be decomposed into the Kronecker product of inverses through $(\mat{A}\otimes \mat{B})^{-1} = \mat{A}^{-1}\otimes \mat{B}^{-1}$. We develop an algorithm to detect and break large tensor inverses into products of smaller tensor inverses so that the computation is cheaper. To keep it simple, the algorithm limits its applicability to specific forms of the tensors, and further details are described in Appendix~\ref{appendix:inverse}. An illustrative example is shown is Figure~\ref{subfig:optinverse}.


\textbf{Redundant node pruning.} We prune the redundant nodes, including the Identity nodes and the inverse nodes, to simplify the expressions. Identity nodes are essential building blocks for the explicit Jacobian and the Hessian expressions, as is shown in Section~\ref{subsec:explicithessian}. During the AD, redundant Identity nodes are introduced to aid the construction of the graph. Hence, we implemented an algorithm to eliminate the unnecessary identity nodes afterwards for better efficiency. Identity nodes are removed unless they express necessary constraints in the output tensor structure, such as the tensor symmetry shown in the right graph of Figure~\ref{subfig:identity}. In addition, we prune the unnecessary inverse nodes, as is shown in Figure~\ref{subfig:pruneinverse}. When there exists an matrix multiplication between an Einsum Node and its corresponding inverse node, we directly return an identity node.

\textbf{Elementary algebraic simplification.} For elementary operations, such as addition, subtraction and multiplication, we use the SymPy library~\cite{meurer2017sympy} to optimize them. SymPy can help us easily simplify the expressions. For the example shown below, it helps reducing the expression to one term.
\begin{python}
sympy_simplify(
    (A-(((A*0.5)-(T*0.5))+((A*0.5)-(T*0.5))))
) = T
\end{python}

\subsection{Optimized Contraction Path Selection}
We identify the optimal contraction path for the Einsum expression after all the above transformations. For one Einsum node with multiple inputs, we provide an function to decompose it into an Einsum graph with the optimized contraction path, as is shown in Figure~\ref{subfig:optpath}.
Our strategy is designed for the common tensor contractions with the following two assumptions:
\begin{itemize}[topsep=0pt,leftmargin=*]
    \item For simplicity, we only discuss the case where tensors are dense, and for a long Einsum expression with multiple inputs, it will first be split into multiple small Einsum expressions, each has only two inputs, and then dense tensor contractions will be executed.
    \item The chosen contraction path is hardware oblivious. We assume the contraction time for each operation is proportional to the flop counts. Other factors, such as the communication cost among different processes under the parallel execution settings, are not considered.
\end{itemize}
These assumptions allow us to implement the algorithm based on 
an interface provided by Opt\_Einsum~\cite{smith2018opt}.
Note that whether we can find the optimal contraction path is based on the optimization algorithm, but we generally found that a greedy search algorithm is able to provide an optimal path for most of the Einsum expressions in tensor computation applications.

In addition, the assumptions above are not limitations of our overall approach. AutoHOOT is also capable of extracting the contraction path based on other libraries, such as Cyclops~\cite{zhang2019enabling}, where hardware and tensor sparsity are considered in the algorithm.

\subsection{Constrained Contraction Path Construction} 
We provide a constrained contraction path selection routine, such that the contraction path is optimized under the constraint that partial inputs' contraction order is fixed. This routine is critical for the dimension tree construction used in the alternating minimization algorithms. Consider Equation~\ref{eq:als}, with the update sequence in each sweep starting from $\tsr{A}_1$ and ending at $\tsr{A}_N$, for the Einsum node used to update $\tsr{A}_i$, where $i\in \{1,\ldots,N\}$, we generate the contraction path such that it is optimized under the constraint that the contraction order for all the target sites is $\tsr{A}_N \prec \cdots \prec \tsr{A}_{i+1} \prec \tsr{A}_1 \prec \cdots \prec \tsr{A}_{i-1}$.
This order ensures that the tensor that is updated just previously, $\tsr{A}_{i-1}$, affects only the last part of the contraction path, enabling the reuse of the calculations prior to it in the path as much as possible.

The constrained path selection algorithm is illustrated in Algorithm~\ref{alg:path_w_constraint}, and is implemented on top of the unconstrained one and uses the greedy search heuristic. We find that this heuristic works well for all the dimension tree selection in the tensor computation applications tested in Section~\ref{sec:bench}. An example is shown in Figure~\ref{subfig:dt}, which illustrate the dimension construction for the Matricized Tensor Times Khatri-Rao Product (MTTKRP) calculations of an order 3 CP decomposition. The pseudo-code is illustrated in Algorithm~\ref{alg:dt} in the Appendix.

\begin{algorithm}[t]
\DontPrintSemicolon
\caption{Opt\_contraction\_path\_w\_constraint}
\label{alg:path_w_constraint}
    \SetAlgoLined
    \SetKwInOut{Input}{input}\SetKwInOut{Output}{output}
    \Input{Einsum Node: N, Contraction order list: L}
    \Output{Einsum Tree: T}
    \SetKwFunction{SplitEinsum}{SplitEinsum}
    \SetKwFunction{Optcontractionpath}{OptContractPath}
    \SetKwFunction{ancestor}{Get\_nearest\_ancestor}
    \SetKwFunction{substitute}{Substitute\_graph}
    $n=$ length(L)

    T = N
    \Comment{Initialize tree with single Einsum node}

    \For{i $\in$ $\{1,\ldots,n\}$} {
        split\_T = \SplitEinsum(T, L[i+1:n]) 
        \Comment{Split T into an Einsum node that contracts all input nodes apart from L[i+1:n] and the subgraph induced by the remaining nodes, returning the former}
        
        opt\_contract\_subtree = \Optcontractionpath(split\_T)
        \Comment{Unconstrained optimized contraction path}
        
        opt\_contract\_subtree = \ancestor(opt\_contract\_subtree, L[i])
        \Comment{Get the tree whose sink is the nearest ancestor of L[i]}
        
        T = \substitute(T, opt\_contract\_subtree)
        \Comment{Return the equivalent graph of T whose inputs contain opt\_contract\_subtree}
    }
    return T
\end{algorithm}

\subsection{Common Subexpression Elimination (CSE)}
CSE is used to remove the duplicated Einsum expressions generated from the path selection above. We show one example in Figure~\ref{subfig:cse}, where CSE helps saving one Einsum calculation. However, CSE is nontrivial for Einsum nodes because different Einsum subscripts may represent the same computation. 
We show an example in Figure~\ref{fig:csediagram} where two Einsum nodes represent the same calculation despite different input ordering and subscripts. Hence, we transfer an Einsum expression into a tensor diagram graph, and compare the graph structures between two expressions. 

Moreover, two nodes in an Einsum graph may be transpositions of each other. After detecting such conditions, we replace one of the nodes with its transpose node and update its outputs' expressions therein. This optimization greatly reduces the computation cost when transposes of large tensors appear in the graph.

\section{Benchmarks}
\label{sec:bench}

\begin{table*}[h]
\small
    \centering
    \begin{tabular}{|c|c|c|c|c|c|c|c|c|}
        \hline
        CPD Kernel & Size ($s$) & Backend & Backend AD & AD & AD + OPT1 & AD + OPT1,2 & AD + OPT1,2,3 & Speed-up \\
    	\hline
        \multirow{2}{*}{GN Jacobian} & \multirow{2}{*}{25} & \multicolumn{1}{c|}{JAX} & \multicolumn{1}{c|}{0.1449s} & \multicolumn{1}{c|}{0.0632s} & \multicolumn{1}{c|}{0.0126s} & \multicolumn{1}{c|}{0.0126s} & 
        \multicolumn{1}{c|}{0.0126s} &  \multicolumn{1}{c|}{11X} 
        \\\cline{3-9}
          &  & \multicolumn{1}{c|}{TensorFlow} & \multicolumn{1}{c|}{1.5201s} & \multicolumn{1}{c|}{0.1037s} & \multicolumn{1}{c|}{0.0029s} & \multicolumn{1}{c|}{0.0029s} & \multicolumn{1}{c|}{0.0029s} & \multicolumn{1}{c|}{524X} 
            \\\hline

        \multirow{4}{*}{GN HVP} & \multirow{2}{*}{40} & \multicolumn{1}{c|}{JAX} & \multicolumn{1}{c|}{0.0107s} & \multicolumn{1}{c|}{0.0011s} & \multicolumn{1}{c|}{0.0012s} & \multicolumn{1}{c|}{0.0011s} & 
        \multicolumn{1}{c|}{0.0011s} &  \multicolumn{1}{c|}{9X} 
        \\\cline{3-9}
          &  & \multicolumn{1}{c|}{TensorFlow} & \multicolumn{1}{c|}{0.0040s} & \multicolumn{1}{c|}{0.0027s} & \multicolumn{1}{c|}{0.0048s} & \multicolumn{1}{c|}{0.0048s} & 
        \multicolumn{1}{c|}{0.0048s} &  \multicolumn{1}{c|}{0.8X}  
          \\\cline{2-9}
         & \multirow{2}{*}{640} & \multicolumn{1}{c|}{JAX} & \multicolumn{1}{c|}{0.3742s} & \multicolumn{1}{c|}{0.776s} & \multicolumn{1}{c|}{0.0056s} & \multicolumn{1}{c|}{0.0054s} & \multicolumn{1}{c|}{0.0051s} & \multicolumn{1}{c|}{73X} 
        \\\cline{3-9}
          &  & \multicolumn{1}{c|}{TensorFlow} & \multicolumn{1}{c|}{0.9669s} & \multicolumn{1}{c|}{0.9746s} & \multicolumn{1}{c|}{0.4470s} & \multicolumn{1}{c|}{0.3422s} & \multicolumn{1}{c|}{0.2795s} & \multicolumn{1}{c|}{3X} 
            \\\hline

        \multirow{4}{*}{ALS Hessian} & \multirow{2}{*}{40} & \multicolumn{1}{c|}{JAX} & \multicolumn{1}{c|}{0.0713s} & \multicolumn{1}{c|}{OOM} & \multicolumn{1}{c|}{0.0017s} & \multicolumn{1}{c|}{0.0017s} & 
        \multicolumn{1}{c|}{0.0017s} &  \multicolumn{1}{c|}{41X} 
        \\\cline{3-9}
          &  & \multicolumn{1}{c|}{TensorFlow} & \multicolumn{1}{c|}{0.3643s} & \multicolumn{1}{c|}{OOM} & \multicolumn{1}{c|}{0.0021s} & \multicolumn{1}{c|}{0.0021s} & \multicolumn{1}{c|}{0.0014s} & \multicolumn{1}{c|}{260X}  
          \\\cline{2-9}
         & \multirow{2}{*}{160} & \multicolumn{1}{c|}{JAX} & \multicolumn{1}{c|}{OOM} & \multicolumn{1}{c|}{OOM} & \multicolumn{1}{c|}{1.0682s} & \multicolumn{1}{c|}{1.0682s} & 
        \multicolumn{1}{c|}{0.8141s} &  \multicolumn{1}{c|}{/} 
        \\\cline{3-9}
          &  & \multicolumn{1}{c|}{TensorFlow} & \multicolumn{1}{c|}{OOM} & \multicolumn{1}{c|}{OOM} & \multicolumn{1}{c|}{3.0164s} & \multicolumn{1}{c|}{3.0164s} & \multicolumn{1}{c|}{1.5405s} & \multicolumn{1}{c|}{/} 
            \\\hline

        \multirow{4}{*}{ALS Hessian inv} & \multirow{2}{*}{40} & \multicolumn{1}{c|}{JAX} & \multicolumn{1}{c|}{0.1623s} & \multicolumn{1}{c|}{OOM} & \multicolumn{1}{c|}{0.0908s} & \multicolumn{1}{c|}{0.0090s} & 
        \multicolumn{1}{c|}{0.0090s} &  \multicolumn{1}{c|}{18X} 
        \\\cline{3-9}
          &  & \multicolumn{1}{c|}{TensorFlow} & \multicolumn{1}{c|}{0.4237s} & \multicolumn{1}{c|}{OOM} & \multicolumn{1}{c|}{0.0278s} & \multicolumn{1}{c|}{0.0028s} & \multicolumn{1}{c|}{0.0028s} & \multicolumn{1}{c|}{151X}  
          \\\cline{2-9}

         & \multirow{2}{*}{160} & \multicolumn{1}{c|}{JAX} & \multicolumn{1}{c|}{OOM} & \multicolumn{1}{c|}{OOM} & \multicolumn{1}{c|}{13.13s} & \multicolumn{1}{c|}{1.5160s} & 
        \multicolumn{1}{c|}{1.5110s} &  \multicolumn{1}{c|}{/} 
        \\\cline{3-9}
          &  & \multicolumn{1}{c|}{TensorFlow} & \multicolumn{1}{c|}{OOM} & \multicolumn{1}{c|}{OOM} & \multicolumn{1}{c|}{OOM} & \multicolumn{1}{c|}{0.5786s} & \multicolumn{1}{c|}{0.5585s} & \multicolumn{1}{c|}{/} 
            \\\hline
    \end{tabular}
    \caption{
    Detailed performance gain from each graph optimization technique on different CPD kernels. 
    The rank is set the same as the input tensor dimension/size along each mode ($s$). Results are collected on an NVIDIA Titan X GPU. We denote each technique as: Einsum fusion + distribution: OPT1, Symbolic optimization: OPT2, CSE: OPT3.
    }
    \label{tab:kernel_comparison}
\end{table*}
\normalsize

\begin{figure*}[]
\centering
\subfloat[KNL CPU]{
\includegraphics[width=.32\textwidth]{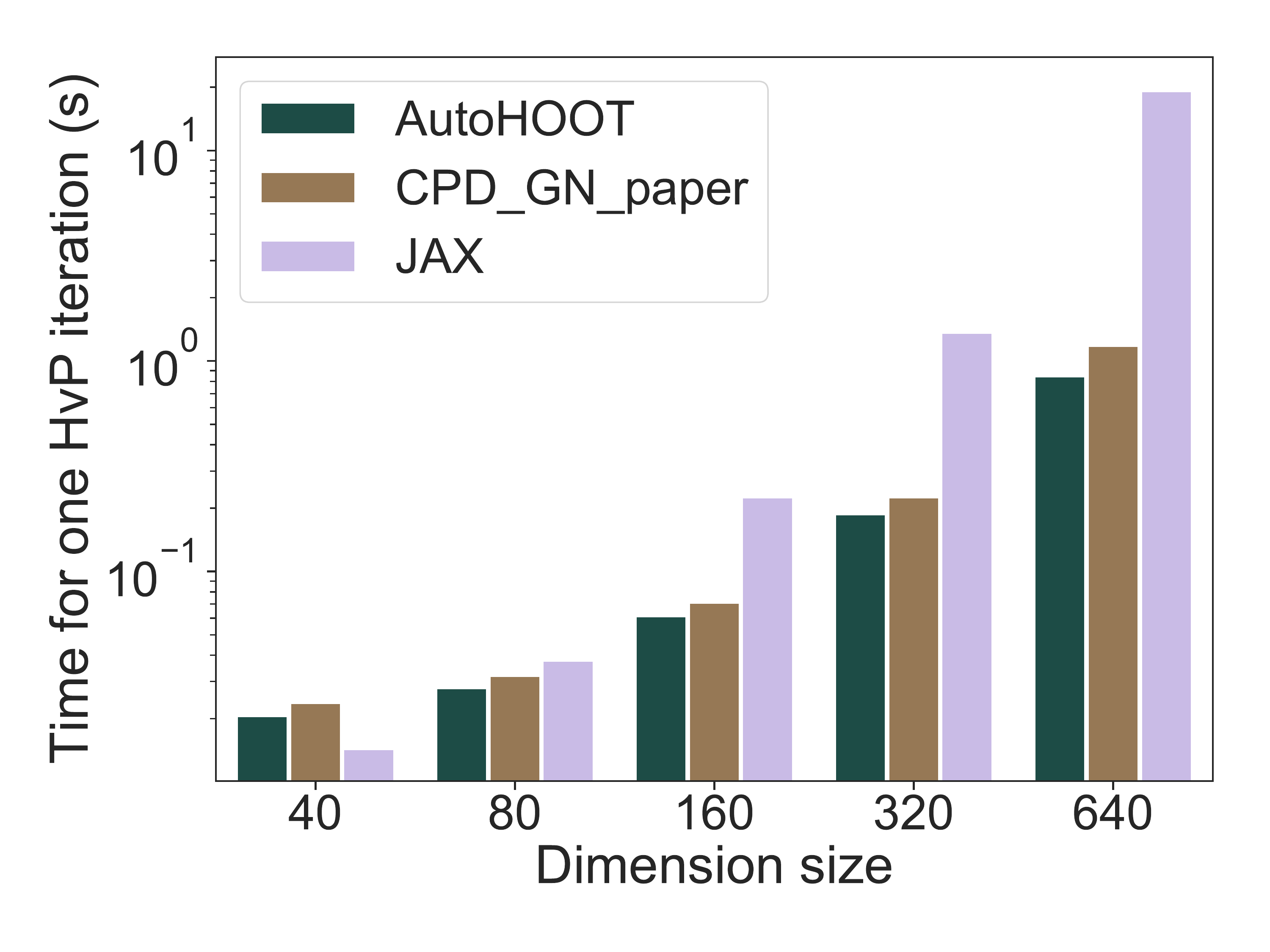}
\label{subfig:cpd-gn-cpu}
}
\subfloat[TESLA K80 GPU]{
\includegraphics[width=.32\textwidth]{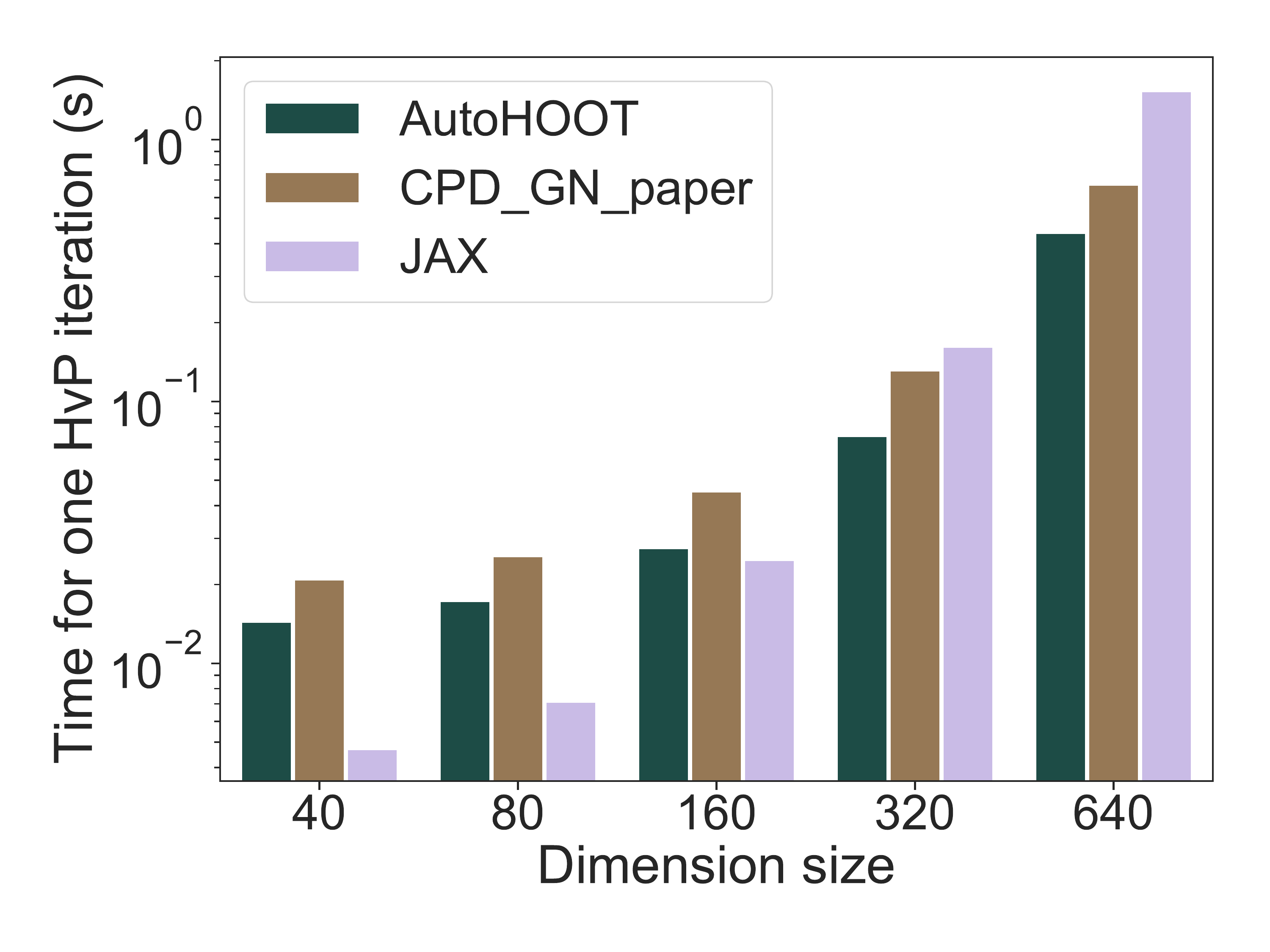}
\label{subfig:cpd-gn-gpu}
}
\caption{Performance comparison among AutoHOOT, JAX and the existing implementation for the HVP kernel in the Gauss-Newton algorithm for the CP decomposition. The implementation of CPD\_GN\_paper comes from reference~\cite{singh2019comparison}. The tensor order is set as $N=3$, and the CP rank is set equal to the dimension size. Each bar is the average result of 10 iterations.
}
\label{fig:cpd-gn}
\end{figure*}

\begin{figure*}[]
\centering
\subfloat[NumPy, CPD, $R=s$]{
\includegraphics[width=.32\textwidth]{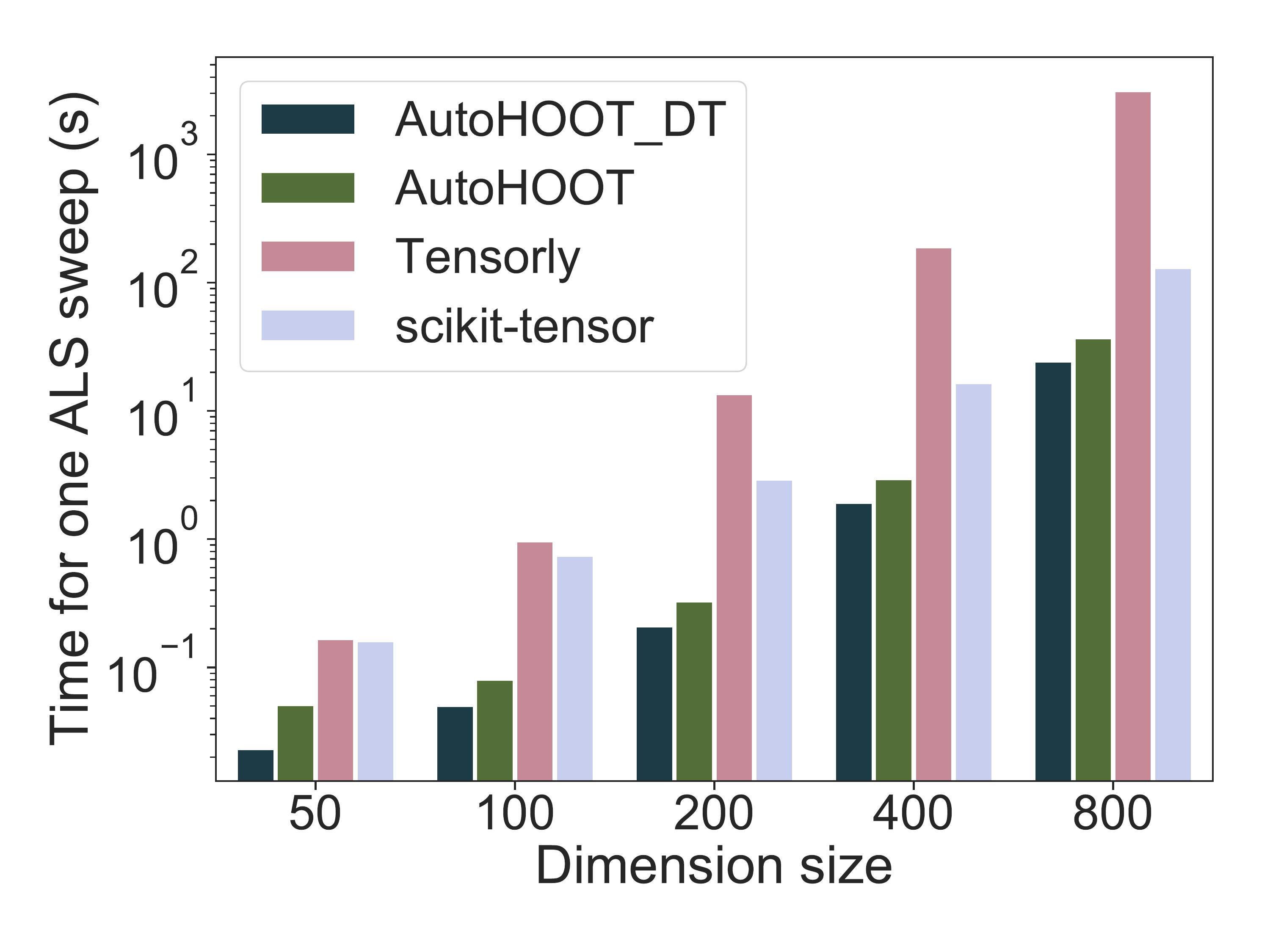}
\label{subfig:cpd-als-cpu}
}
\subfloat[TensorFlow, CPD, $R=s$]{
\includegraphics[width=.32\textwidth]{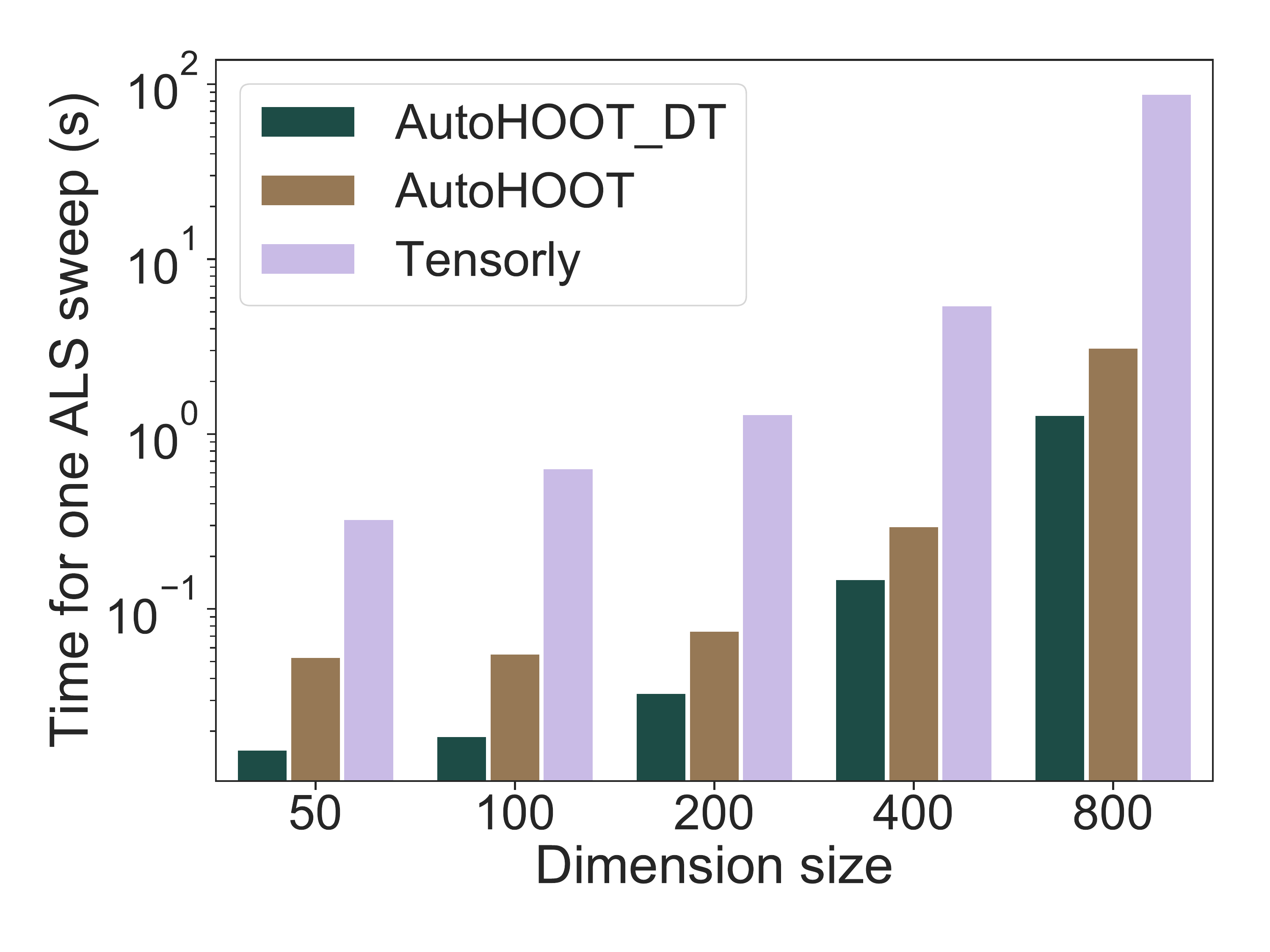}
\label{subfig:cpd-als-gpu}
}
\subfloat[Cyclops, CPD, $R=400, s=\lfloor 1320 n^{\frac{1}{3}}\rfloor$]{
\includegraphics[width=.32\textwidth]{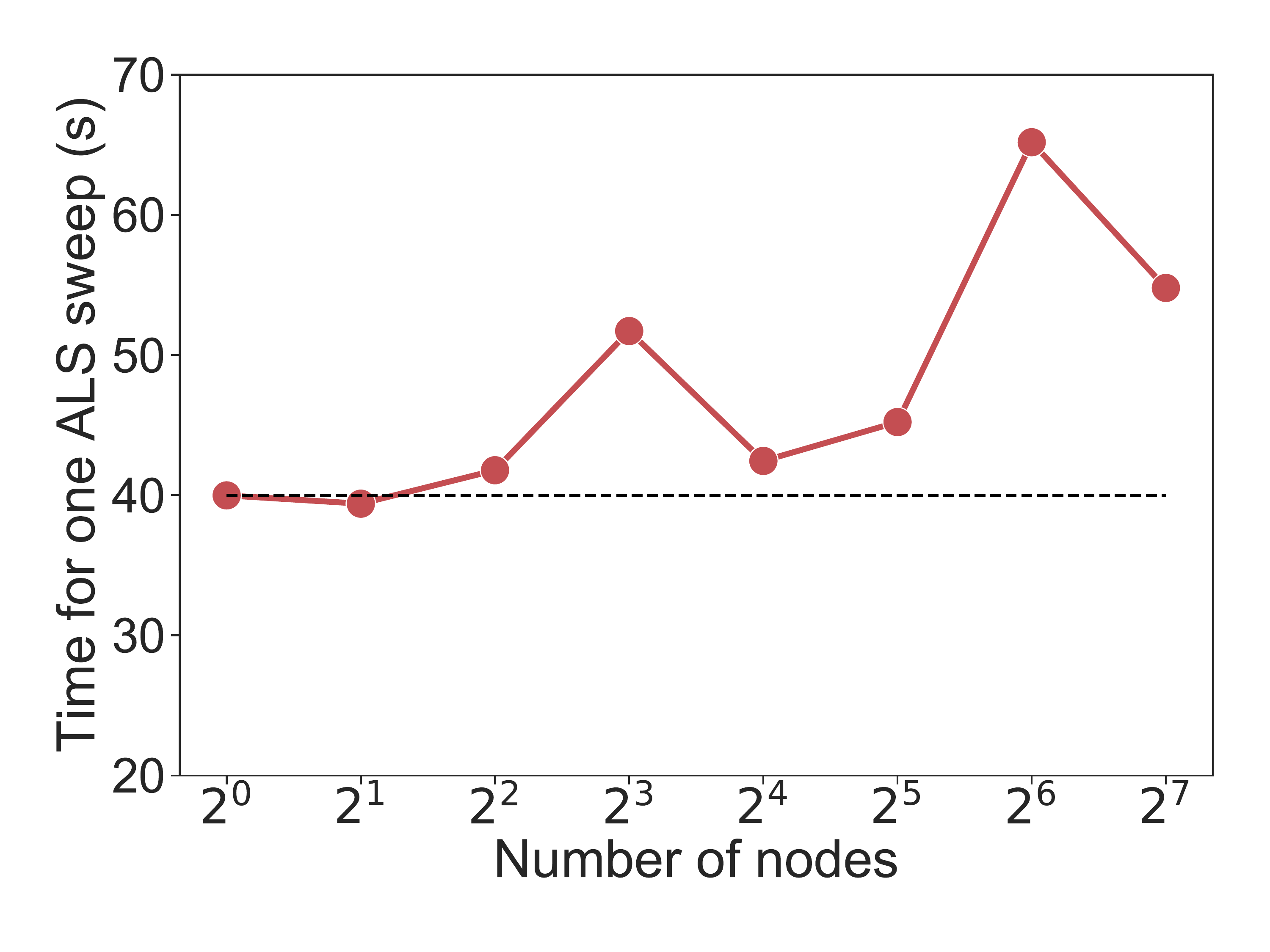}
\label{subfig:cpd-als-ctf}
}

\subfloat[NumPy, Tucker, $R=0.5s$]{
\includegraphics[width=.32\textwidth]{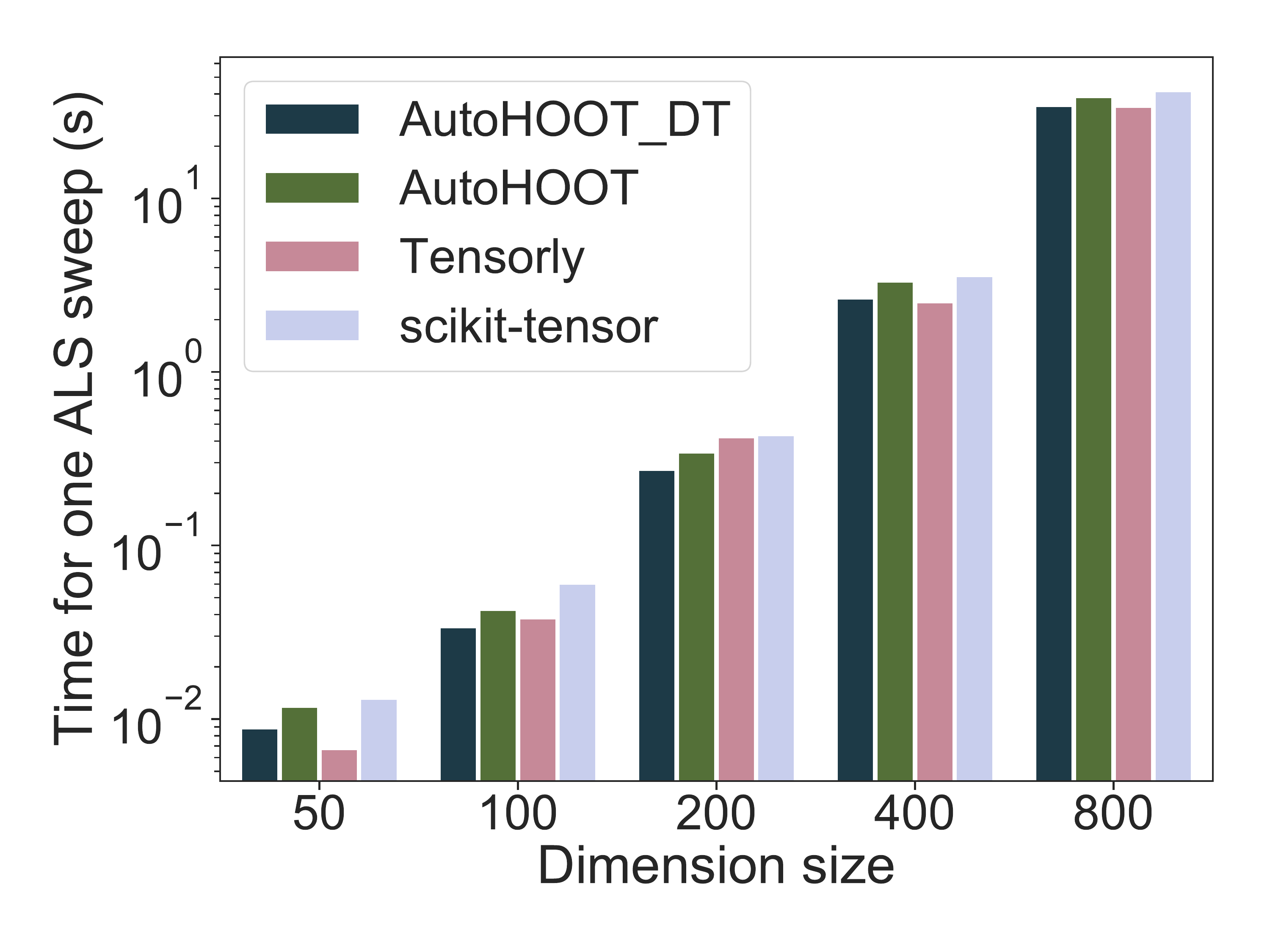}
\label{subfig:tucker-als-cpu}
}
\subfloat[TensorFlow, Tucker, $R=0.5s$]{
\includegraphics[width=.32\textwidth]{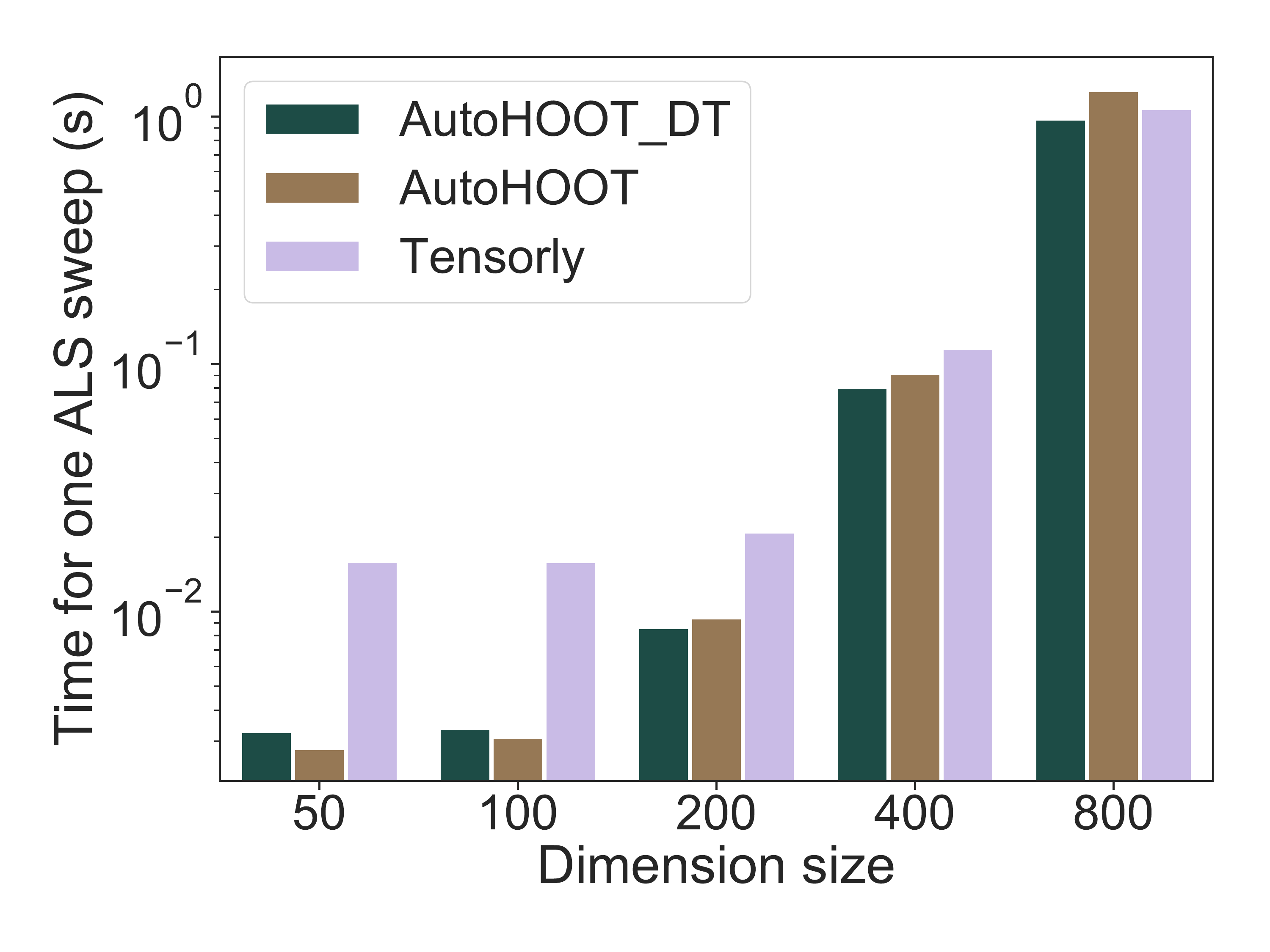}
\label{subfig:tucker-als-gpu}
}
\subfloat[Cyclops, Tucker, $R=400, s=\lfloor 1240 n^{\frac{1}{3}}\rfloor$]{
\includegraphics[width=.32\textwidth]{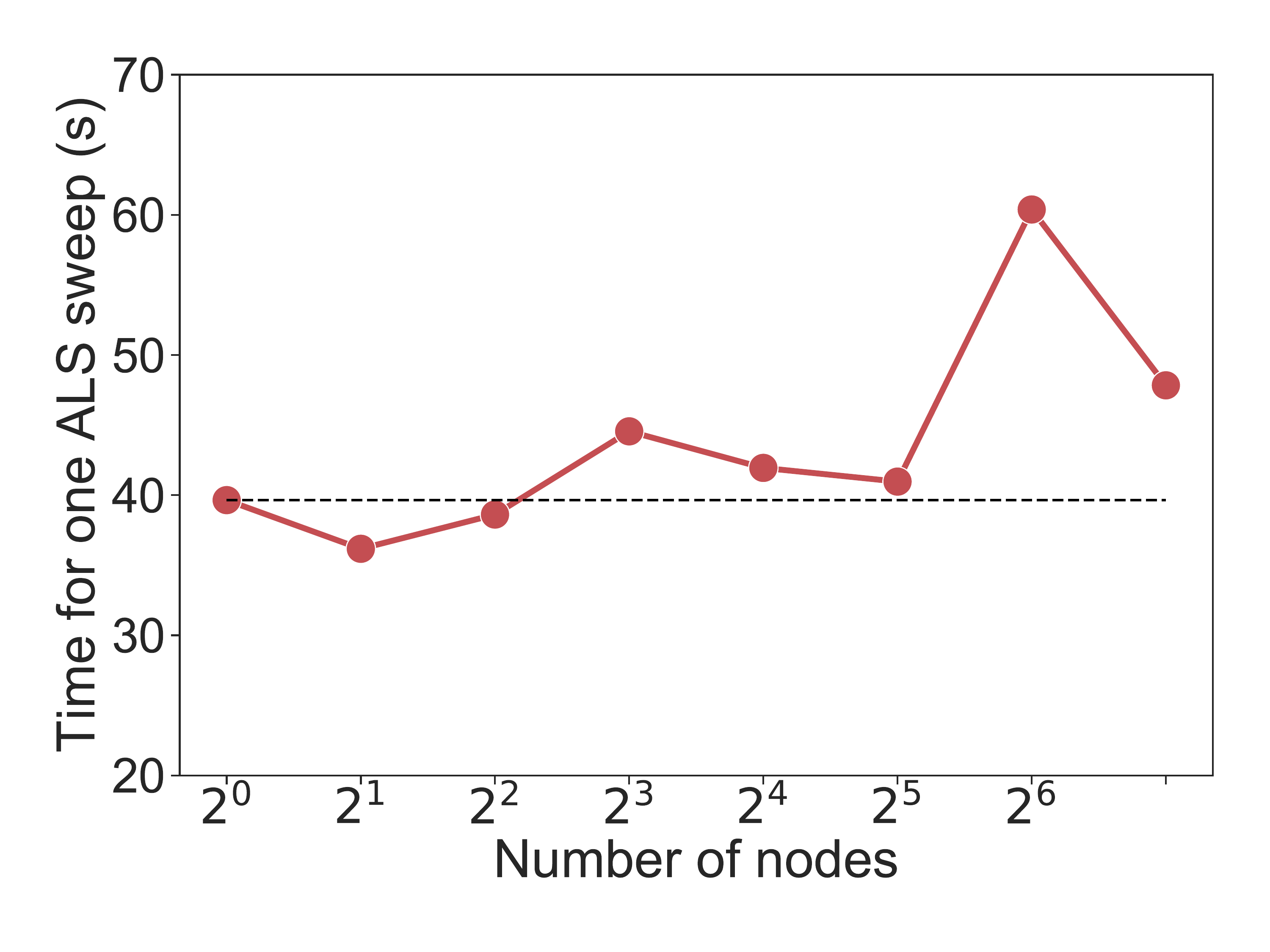}
\label{subfig:tucker-als-ctf}
}

\subfloat[NumPy, DMRG, $s=R$]{
\includegraphics[width=.32\textwidth]{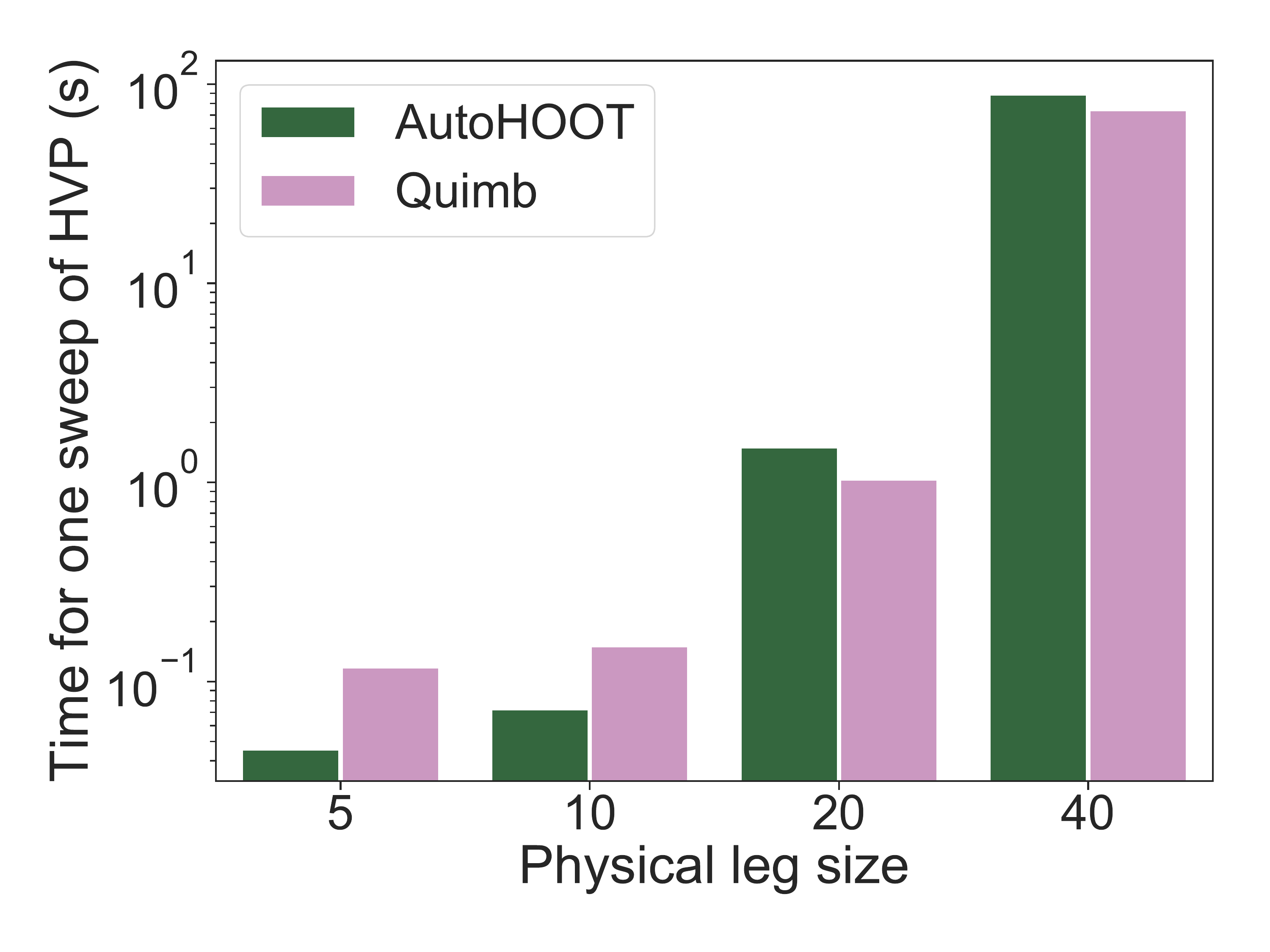}
\label{subfig:dmrg-als-cpu}
}
\subfloat[TensorFlow, DMRG, $s=R$]{
\includegraphics[width=.32\textwidth]{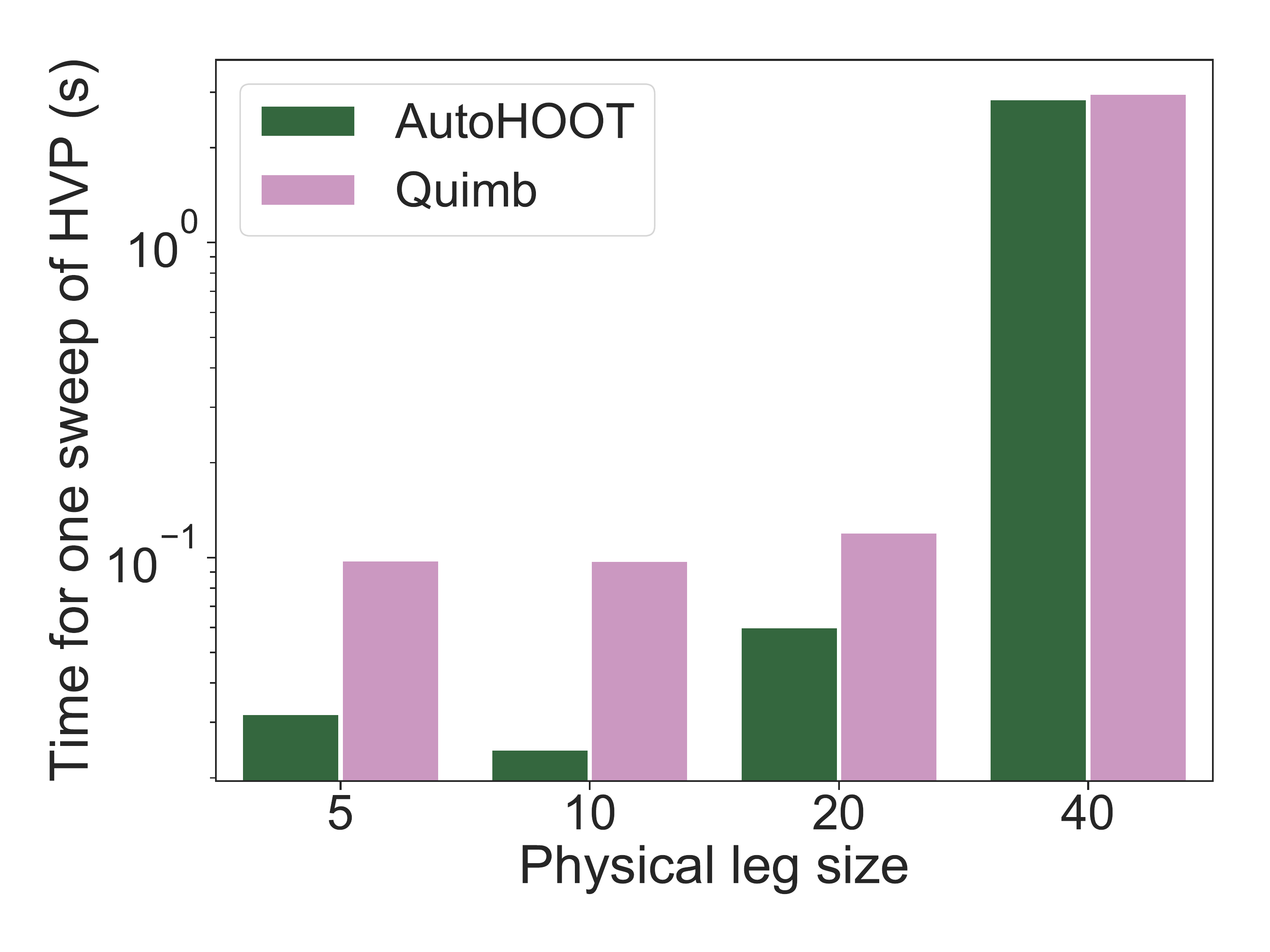}
\label{subfig:dmrg-als-gpu}
}
\subfloat[Cyclops, DMRG, $R=s=\lfloor 50n^{\frac{1}{5}} \rfloor$]{
\includegraphics[width=.32\textwidth]{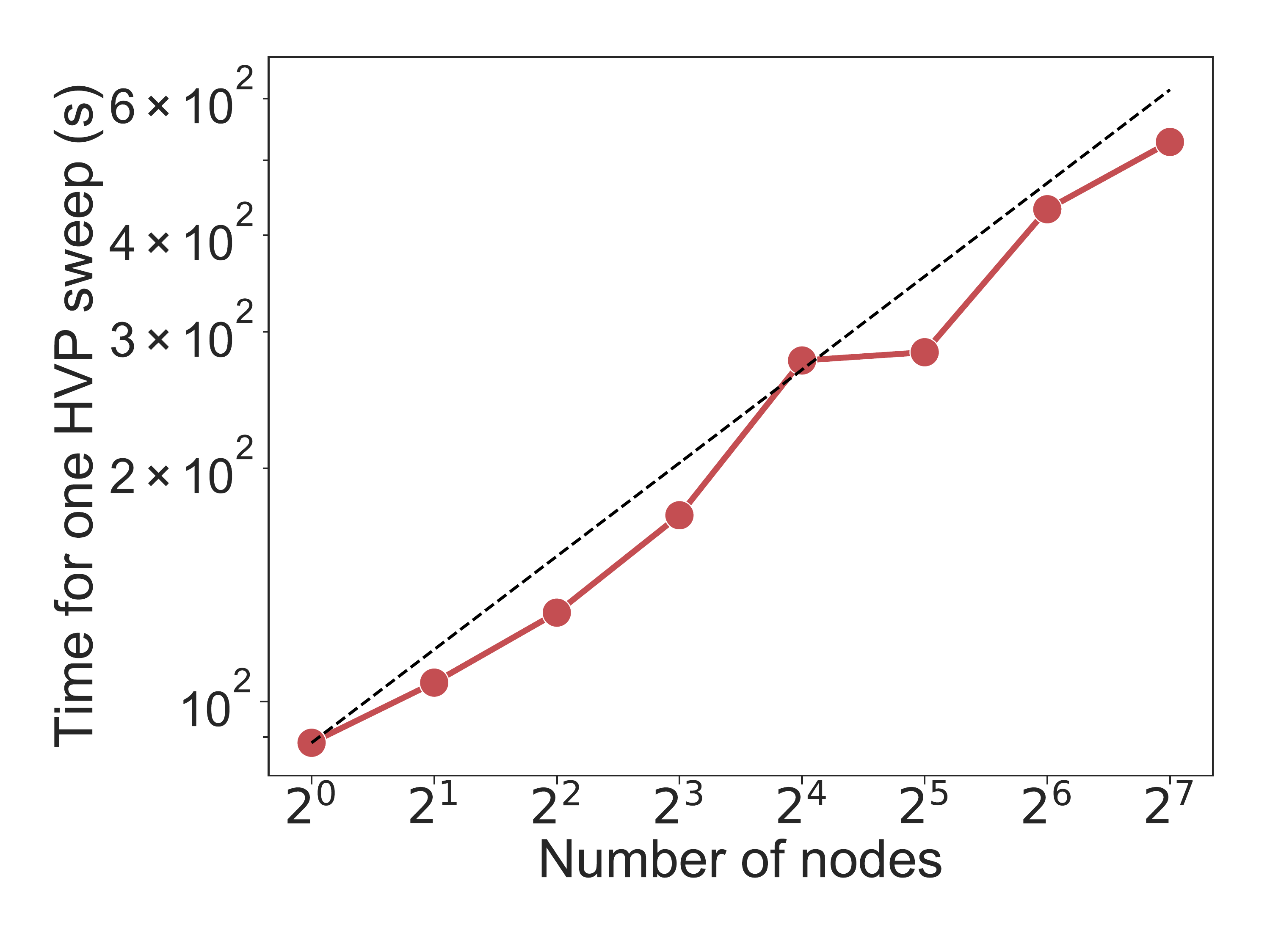}
\label{subfig:dmrg-als-ctf}
}
\caption{AutoHOOT performance for kernels in the alternating minimization. 
\textbf{(a)-(c)}: Results for the CP decomposition. The tensor order is set as $N=3$ for all the experiments. 
\textbf{(d)-(f)}: Results for the Tucker decomposition. The tensor order is set as $N=3$ for all the experiments. 
\textbf{(g)-(i)}: Results for the DMRG experiment. The number of sites is set as $N=10$ for the experiments with NumPy and TensorFlow, and set as $N=6$ for the experiments with Cyclops. 
For the Cyclops benchmark, the dotted line denotes the perfect scaling curve. Each bar/dot is the average result of 10 iterations. 
}
\label{fig:als}
\end{figure*}

\begin{figure}[]
\centering
\includegraphics[width=.5\textwidth]{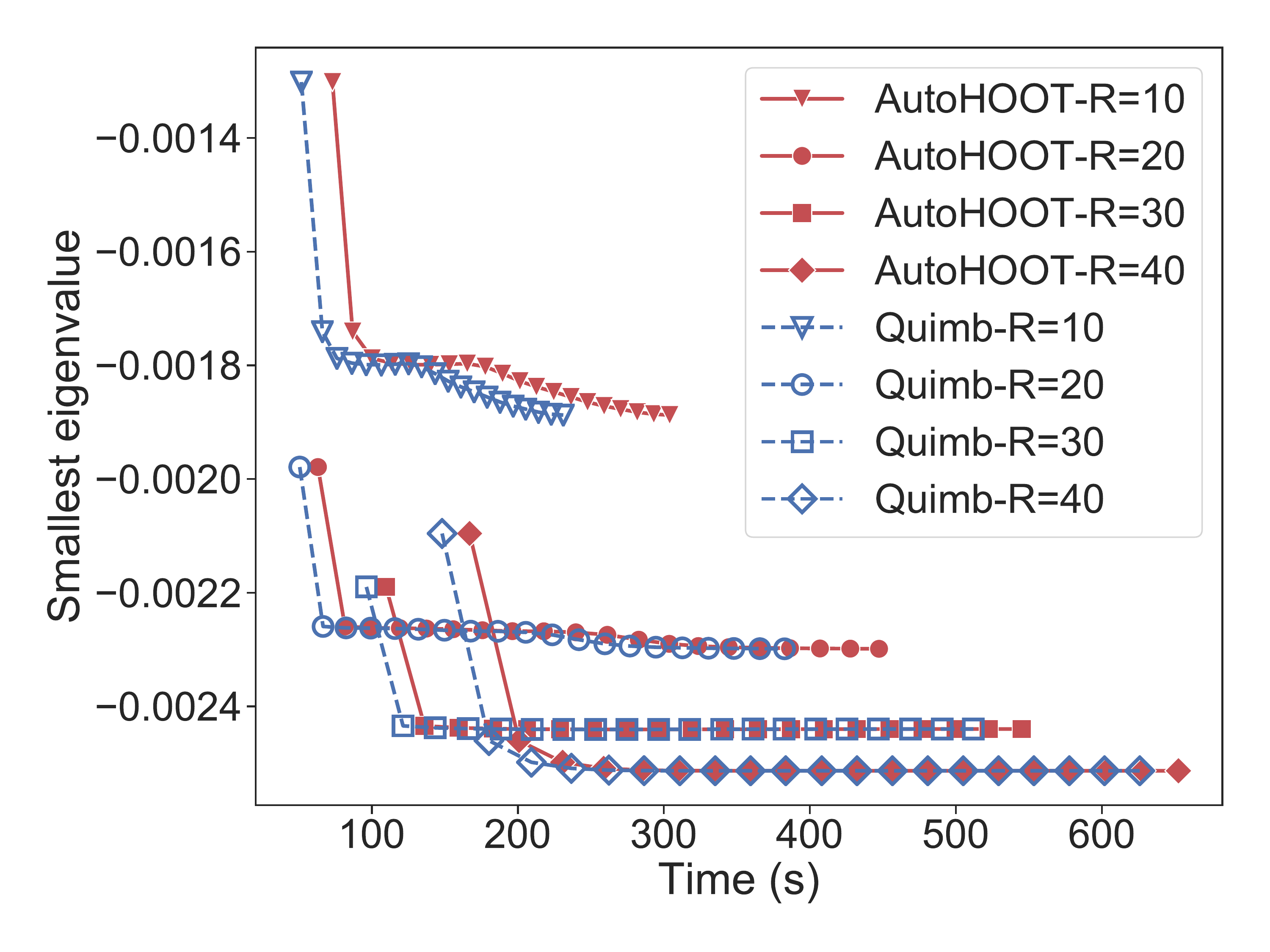}
\caption{Comparison between AutoHOOT and Quimb on the full DMRG running curve. The input MPO is random and symmetric, has 6 sites, and its physical leg size equals 10 and MPO rank size equals 20. We compare the performance under different largest MPS rank constraints.}
\label{fig:dmrgfull}
\end{figure}

We evaluate the performance of AutoHOOT on both the Gauss-Newton method and the alternating minimization method discussed in Section~\ref{subsec:numericalopt}. The performance of the critical Gauss-Newton kernel, the Hessian-Vector Product, is evaluated on the CP decomposition application, where Gauss-Newton with conjugate gradient update is commonly used to achieve high accuracy~\cite{singh2019comparison,sorber2013optimization}. The performance of alternating minimization kernels generated by AutoHOOT is evaluated on both CP and Tucker decompositions, as well as the DMRG algorithm in tensor network applications used to calculate the smallest eigenvalue and eigenvector for a matrix product state.

The experiments are run on both CPUs and GPUs. On CPUs, we test the performance on both one process with the NumPy backend, and on the distributed parallel system with the Cyclops backend. The results are collected on the Stampede2 supercomputer located at the University of Texas at Austin. We leverage the Knight’s Landing (KNL) nodes, each of which consists of 68 cores, 96 GB of DDR RAM, and 16 GB of MCDRAM. These nodes are connected via a 100 Gb/sec fat-tree Omni-Path interconnect. We use Intel compilers and the MKL library for threaded BLAS routines for both sequential and parallel experiments. We use 16 processes per node and 16 threads per process for the Cyclops benchmark experiments. We also collected results with both TensorFlow and JAX backends on both single NVIDIA TESLA K80 GPU and single NVIDIA Titan X GPU.

We first compare the detailed performance gain from each graph optimization technique proposed in Section~\ref{sec:opt}. The experiments are performed on the Jacobians and HVPs kernels in the Gauss-Newton (GN) methods, as well as Hessians and Hessian inverses used in the ALS algorithm for CP decompositions and are shown in Table~\ref{tab:kernel_comparison}. As can be seen in the table, Einsum fusion and distribution are critical for almost all the calculations, and Symbolic optimization is critical for tensor/matrix inverse. In addition, CSE provides incremental performance gain.

The performance of the HVP kernels in the Gauss-Newton algorithm for the CP decomposition is shown in Figure~\ref{fig:cpd-gn}. As can be seen, AutoHOOT has at least 2X speed-up on the GPU and at least 7X speed-up on the CPU compared to JAX when the dimension size $s\geq 320$. 
Note that JAX performs better for small HVP kernels, because the experiments with AutoHOOT are performed on TensorFlow, where JAX has faster small contractions.
It can be seen that the speed-up increases with the increase of the dimension size, indicating the advantage of AutoHOOT for large scale tensor computations. In addition, the AutoHOOT performance is comparable compared to the manually designed algorithms in the reference~\cite{singh2019comparison}, indicating that the kernels generated by AutoHOOT reaches the state-of-art performance boundary.

The performance of the alternating minimization kernels for both tensor decompositions and the DMRG algorithm are shown in Figure~\ref{fig:als}. For the tensor decompositions, we compare the performance of AutoHOOT output expressions, both with and without dimension tree optimizations, to the popular tensor decomposition libraries Tensorly~\cite{kossaifi2019tensorly}, both with NumPy and TensorFlow backend, and scikit-tensor\footnote{\url{https://github.com/mnick/scikit-tensor}} with NumPy backend. For the DMRG algorithm, we compare the performance to Quimb~\cite{gray2018quimb}, which is an efficient library for tensor networks. 

The benchmark results for the CP decomposition with both NumPy and TensorFlow can be seen in Figure~\ref{subfig:cpd-als-cpu}, \ref{subfig:cpd-als-gpu}. 
We compare the performance with different CP ranks ($R$) and dimension size ($s$).
As can be seen, the expressions generated with the dimension tree algorithm outperform all the other implementations.
Note that Tensorly's performance is not as expected for the CP decomposition, because it slices the factor matrices over the rank mode and sums over all the MTTKRP results of the input tensor and the sliced factor matrices, which is not favorable.
The weak scaling benchmark is also performed on the distributed parallel system with Cyclops, shown in Figure~\ref{subfig:cpd-als-ctf}, where we consider weak scaling with fixed input size and work per processor.
The expressions generated from AutoHOOT 
scale well, obtaining 73\% parallel scaling efficiency on 128 nodes (2048 cores).

The benchmark results for the Tucker decomposition with both NumPy and TensorFlow can be seen in Figure~\ref{subfig:tucker-als-cpu}, \ref{subfig:tucker-als-gpu}. 
We compare the performance with different Tucker ranks ($R$) and dimension size ($s$).
Note that we are only comparing the performance of the kernel generated through AutoHOOT to the Tensor Times Matrix-chain (TTMc) implementation in other libraries, which doesn't contain the low rank factorization step of splitting the factor matrix from the core tensor. The expressions generated with the dimension tree algorithm achieve comparable performance to all the other implementations. The weak scaling benchmark is shown in Figure~\ref{subfig:tucker-als-ctf}. Similar to the CP decomposition, the expressions generated from AutoHOOT scale with high efficiency.

The performance results for DMRG can be seen in Figure~\ref{subfig:dmrg-als-cpu}, \ref{subfig:dmrg-als-gpu}, \ref{subfig:dmrg-als-ctf}. We benchmark over sweep of the HVP kernels with different MPO and MPS rank size ($R$) and physical dimension size ($s$),
where the Hessian denotes the local Hessian of the DMRG loss function w.r.t. each local site. In DMRG, the HVP calculations are important kernels for the sparse eigensolver. Multiple HVP calculations are necessary for each site to get the local smallest eigenvalue, making it the computation bottleneck.
The expressions generated with the dimension tree algorithm achieve comparable performance to the implementations in Quimb. In addition, the expressions generated from AutoHOOT scale nearly perfectly with Cyclops up to at least 128 nodes\footnote{In these experiments, we constrain the physical leg size to be equal to the rank, e.g. $s=R$, so the computational cost is $O(R^7)$ and the memory footprint is $O(R^5)$.}. 

We also compare the performance between AutoHOOT and Quimb on the full DMRG experiments. Like Quimb, we use the sparse eigensolver in SciPy~\cite{2020SciPy-NMeth}, and set the solver parameters the same as Quimb. The results are shown in Figure~\ref{fig:dmrgfull}. We test the four cases where the maximum MPS rank ranges from 10 to 40, and the results show that both libraries have the similar performance, while AutoHOOT has a small fixed overhead. 

Note that we did not report the ALS results of other AD libraries, because their performance is far worse than both AutoHOOT and other tensor computation libraries. For both CP and Tucker decompositions, existing AD libraries cannot efficiently decompose the structured inverse operations, leading to a big overhead from inverting large tensors. For the DMRG experiment, existing libraries fail to choose an optimized contraction path, and produce large intermediates which require too much memory.

\section{Conclusion}
\label{sec:conclu}

AutoHOOT is the first automatic differentiation framework targeting high-order optimization for tensor computations.
AutoHOOT contains a new explicit Jacobian / Hessian expression generation kernel whose outputs keep the input tensors' granularity and are easy to optimize. 
It also contains a new computational graph optimization module that combines both the traditional optimization techniques for compilers and techniques based on specific tensor algebra. The optimization module generates expressions as good as manually written codes in other frameworks for the numerical algorithms of tensor computations.
AutoHOOT is compatible with other numerical computation libraries, and users can execute the generated expressions on CPU with NumPy, GPU with TensorFlow, and distributed parallel systems with Cyclops Tensor Framework.
Experimental results show that AutoHOOT has competitive performance on both tensor decomposition and tensor network applications compared to both existing AD software and other tensor computation libraries with manually written kernels, both on CPU and GPU architectures. 

\section{Acknowledgements}
Linjian Ma and Edgar Solomonik were supported by the US NSF OAC SSI program, award No.\ 1931258.
This work used the Extreme Science and Engineering Discovery Environment (XSEDE), which is supported by National Science Foundation grant number ACI-1548562.
We used XSEDE to employ Stampede2 at the Texas Advanced Computing Center (TACC) through allocation TG-CCR180006.

\bibliographystyle{abbrv}
\bibliography{main}

\clearpage
\label{sec:appendix}

\onecolumn

\appendix
\section{Appendix}

\subsection{Background of Tensor Computation Applications}
In this Section we provide background on CP decomposition, Tucker decomposition, and Density Matrix Renormalization Group (DMRG). 

\textbf{CP decomposition.}
The CP tensor decomposition~\cite{hitchcock1927expression,harshman1970foundations} serves to approximate a tensor by a sum of $R$ tensor products of vectors. 
For an order $N$ input tensor $\tsr{X}$ with size $s_1\times\cdots\times s_N$, CP decomposition compresses it into $N$ factor matrices $\mat{A}_1, \ldots, \mat{A}_N$, size of each is $s_i \times R$ for $i\in\{1,\ldots, N\}$. 
The optimization for the CP decomposition is a least squares problem, where element-wise expression for the output of the tensor network $f$ in Equation~\ref{eq:leastsquare} denotes
\begin{align*}
f(\mat{A}_1, \ldots, \mat{A}_N)(x_1, \ldots, x_N) :=  \sum_{k=1}^{R}\prod_{i\in\{1,\ldots,N\}} \mat{A}_{i}(x_i,k).
\end{align*}
Both the Gauss-Newton method and the alternating minimization method, which is also called alternating least squares (ALS) are effective and popular for the CP decomposition.

The CP-ALS method alternates among quadratic optimization problems for each of the factor matrices $\mat{A}_{n}$, resulting in linear least squares problems for each row,  are often solved via the normal equations~\cite{kolda2009tensor},
\[
 \mat{A}_{n}\boldsymbol{\Gamma}_{n}\leftarrow \mat{X}_{(n)}\mat{P}_{n},
\]
where the matrix $\mat{P}_{n}\in \mathbb{R}^{I_n \times R}$, where $I_n = \prod_{i=1, i\neq n}^{N} s_i$, is formed by Khatri-Rao products of the other factor matrices,
\[
    \mat{P}_{n}=\mat{A}_{N} \odot \cdots \odot  \mat{A}_{n+1}  \odot  \mat{A}_{n-1} \odot \cdots \odot \mat{A}_{1}, 
\]
and $\mat{\Gamma}_{n}\in\mathbb{R}^{R\times R}$ can be computed via
\[
\boldsymbol{\Gamma}_{n}=\mat{S}_{1}\ast\cdots\ast \mat{S}_{n-1}  \ast \mat{S}_{n+1}\ast\cdots\ast \mat{S}_{N},
\label{eq:gamma}
\]
with each $\mat{S}_{i} = \mat{A}_{i}^{T}\mat{A}_{i}.$
The \textit{Matricized Tensor Times Khatri-Rao Product} or \mbox{MTTKRP} computation $\mat{M}_{n}=\mat{X}_{(n)}\mat{P}_{n}$ is the main computational bottleneck of CP-ALS~\cite{ballard2017communication}.
Within MTTKRP,  the bottleneck is the contraction between the input tensor and the first-contracted matrix.
For a rank-R CP decomposition, this computation has the leading cost of $2s^NR$ if $s_n=s$ for all $n\in\inti{1}{N}$.
The dimension-tree algorithm for ALS~\cite{phan2013fast,vannieuwenhoven2015computing} uses a fixed amortization scheme to update \mbox{MTTKRP} in each ALS sweep. 
This scheme only needs to perform two first contraction calculations for each ALS sweep, decreasing the leading order cost of a sweep from $2Ns^{N}R$ to  $ 4s^{N}R. $ 

Another alternative is to solve the least squares problem via the Gauss-Newton method. Although directly inverting the Hessian matrix for the problem is expensive (costs $O(N^3s^3R^3)$ if each dimension has size $s$), the matrices involved in the linear system are sparse and have much implicit structure. The cost of direct Hessian inversion can be reduced to $O(NR^6)$~\cite{tichavsky2013further} and of using implicit conjugate gradient method is $O(N^2sR^2)$ for each iteration~\cite{sorber2013optimization}. Additionally, it has been shown that higher decomposition accuracy can be reached with Gauss-Newton rather than ALS~\cite{singh2019comparison}. We refer readers to references for details of the Gauss-Newton implementations for CP decomposition~\cite{singh2019comparison,sorber2013optimization}. 

\textbf{Tucker decomposition.}
Tucker decomposition~\cite{tucker1966some} approximates a tensor by a core tensor contracted by orthogonal matrices along each mode. 
For an order $N$ input tensor $\tsr{X}$ with size $s_1\times\cdots\times s_N$, Tucker decomposition compresses it into $N$ matrices with orthogonal columns $\mat{A}_1, \cdots, \mat{A}_N$, size of each is $s_i \times R_i$ for $i\in\{1,\ldots, N\}$, and a core tensor $\tsr{G}$ with size $R_1\times\cdots\times R_N$. 
Similar to CP decomposition, the optimization for the Tucker decomposition is a least squares problem, where element-wise expression for the output of the tensor network $f$ in Equation~\ref{eq:leastsquare} denotes
\begin{align*}
f(\tsr{G}, \mat{A}_1, \ldots, \mat{A}_N)(x_1,\ldots,x_N) =  
\sum_{\{z_1,\ldots, z_N\}}\tsr{G}(z_1,\ldots,z_N)\prod_{r\in\{1,\ldots,N\}} \mat{A}_{r}(x_r,z_r).
\end{align*}
The ALS method for Tucker decomposition~\cite{andersson1998improving,kolda2009tensor}, which is also called the \textit{higher-order orthogonal iteration} (HOOI), proceeds by fixing all except one factor matrix, and computing a low-rank matrix factorization on the \textit{Tensor Times  Matrix-chain} (TTMc) $\tsr{Y}_{n}$ for $n\in\{1,\ldots,N\}$, to update that factor matrix and the core tensor. $\tsr{Y}_{n}$ is expressed as
    \begin{align*}
    \tsr{Y}_{n}(z_1,\ldots,z_{n-1},x_{n},z_{n+1},\ldots,z_{N})=
\sum_{\{x_1,\ldots,x_{n-1},x_{n+1},\ldots x_N\}}\tsr{X}(x_1,\ldots,x_N)\prod_{r\in\{1,\ldots,N\}, r\neq n} \mat{A}_{r}(x_r,z_r).
    \end{align*}
\normalsize
Then $\tsr{Y}_{n}$ is factored into a product of an orthogonal matrix $\mat{A}_{n}$ and the core tensor $\tsr{G}$, so that
    \(
    \mat{Y}_{n,(n)} \approx \mat{A}_{n}\mat{G}_{(n)} .
    \)
This factorization can be done by taking $\mat{A}_{n}$ to be the $R_n$ leading left singular vectors of $\mat{Y}_{n,(n)}$. TTMc is the computational bottleneck of Tucker-ALS. With the use of dimensions trees same as CP-ALS, the  computational complexity  for  a sweep of TTMc has the leading  order $O(4s^NR)$.

\textbf{Density Matrix Renormalization Group (DMRG).}
DMRG calculates the smallest eigenvalue of a Matrix Product Operator (MPO) with the corresponding eigenvector represented by a Matrix Product State (MPS).
MPS, which is also called tensor train, represents a high dimensional tensor into a linear tensor network. For an order $N$ input tensor $\tsr{V}$ with size $s_1\times\cdots\times s_N$, the MPS decomposition is expressed as
\[
\tsr{V}(x_1, \ldots, x_N) = \sum_{\alpha_0, \ldots, \alpha_N}\prod_{i=1}^N \tsr{A}_{i}(\alpha_{i-1}, x_i, \alpha_{i}),
\]
where $\tsr{A}_{i}\in\R^{R_{i-1}\times s_i \times R_{i}}$ and $R_0=R_N=1$. 
The MPO has the similar linear structure, each site is a 4-D tensor,
\[
\tsr{W}(x_1, \ldots, x_N, y_1, \ldots, y_N) = \sum_{\alpha_0, \ldots, \alpha_N}\prod_{i=1}^N \tsr{A}_{i}(\alpha_{i-1}, x_i, y_i, \alpha_{i}),
\]
and for each $i\in\{1,\ldots,N\}$, the $i$th and $i+N$th mode of $\tsr{W}$ have the same size.
The objective of DMRG is expressed as
\[
\min_{\tsr{V}}\psi(\tsr{V}) := \frac{
\vcr{v}^T_{(1:N)} \mat{W}_{(1:N)}\vcr{v}_{(1:N)}
}
{\|\vcr{v}_{(1:N)}\|^2
},
\]
where we are optimizing $\tsr{V}$ under the constraint that it has the MPS structure. 
DMRG optimizes this objective via alternating minimization, where in each local step the minimum of the objective w.r.t. one or two neighboring sites is achieved, and performs sweeps of the local steps until the results converged.
We refer readers to the tensornetwork website for algorithm details\footnote{https://tensornetwork.org/mps/algorithms/dmrg/}.

\vspace{4mm}

\subsection{Proofs for Structured Inverse Node Decomposition}\label{appendix:inverse}

In our program, for an implicit tensor constructed through several input tensors and an Einsum expression, our optimization algorithm finds the form of its decomposed tensors that obey the rules in Theorem~\ref{thm:inv}, thus helping the inverse.

At first, we define the terms \textit{decomposable tensor}, \textit{tensor inverse} and \textit{identity tensor} as follows:
\begin{defn}
\label{defn:decompose}
A tensor $\tsr{T} \in\mathbb{R}^{s_1\times \cdots \times s_N}$ is \textit{decomposable} if it can be written as the outer product of 2 smaller tensors. It can be written as
\[
\tsr{T}(x_1, \ldots, x_N) = \tsr{A}(y_1, \ldots, y_M)\tsr{B}(z_1, \ldots, z_K),
\]
where $\{y_1, \ldots, y_M\}\cup \{z_1, \ldots, z_K\} = \{x_1,\ldots, x_N\}$ and  $\{y_1, \ldots, y_M\}\cap \{z_1, \ldots, z_K\} = \emptyset$.
\end{defn}

\begin{defn}
\label{defn:inv}
For an even order tensor $\tsr{T} \in\mathbb{R}^{s_1\times \cdots \times s_{2N}}$,
let $R_1=\prod_{i=1}^{N}s_i, R_2=\prod_{i=N+1}^{2N}s_i$, if $R_1=R_2$,
its \textit{tensor inverse} $\tsr{T}^{-1}$ is defined as the inverse of the matricized tensor $\mat{T}$, where $\mat{T}\in \mathbb{R}^{R_1\times R_2}$.
\end{defn}
\begin{defn}
We use $\tsr{I}_T$ to denote a tensor has the same shape as $\tsr{T}$, and the matricized $\tsr{I}_T$ is an identity matrix. 
\end{defn}
Using the definitions above, we will show that if a tensor meets the requirement described in Theorem~\ref{thm:inv}, then we can transfer the tensor inverse into the inverse of its decomposed parts.
\begin{thm}
\label{thm:inv}
For an even order tensor $\tsr{T} \in\mathbb{R}^{s_1\times \cdots \times s_{2N}}$,
if it can be decomposed into 2 tensors $\tsr{A}$ and $\tsr{B}$ as:
\[
\tsr{T}(x_1, \ldots, x_{2N}) = \tsr{A}(y_1, \ldots, y_{2M})\tsr{B}(z_1, \ldots, z_{2K}),
\]
and the indices satisfy the following requirements:
\begin{enumerate}
    \item $\{y_1, \ldots, y_M\}\cup \{z_1, \ldots, z_K\} = \{x_1,\ldots, x_N\}$, and  $\{y_1, \ldots, y_M\}\cap \{z_1, \ldots, z_K\} = \emptyset$,
    \item $y_{i+M} = x_{j+N}$ if $y_{i} = x_{j}$ for $\forall i\in\{1,\ldots, M\}$, $z_{i+K} = x_{j+N}$ if $z_{i} = x_{j}$ for $\forall i\in\{1,\ldots, K\}$, 
    \item $\tsr{A}$ and $\tsr{B}$ are both invertible,
\end{enumerate}
then we have 
$
\tsr{T}^{-1}(x_1, \ldots, x_{2N}) = \tsr{A}^{-1}(y_1, \ldots, y_{2M})\tsr{B}^{-1}(z_1, \ldots, z_{2K}).
$
\end{thm}
\begin{proof}
Let tensor $\tsr{C}$ be the outer product of tensor $\tsr{A}^{-1}$ and $\tsr{B}^{-1}$ based on the following element-wise expression:
\[
\tsr{C}(x_1, \ldots, x_{2N}) = \tsr{A}^{-1}(y_1, \ldots, y_{2M})\tsr{B}^{-1}(z_1, \ldots, z_{2K}),
\]
we can rewrite the equation above with different index notations:
\[
\tsr{C}(x_{N+1}, \ldots, x_{2N}, a_1, \ldots, a_N) = \tsr{A}^{-1}(y_{M+1}, \ldots, y_{2M},y_{2M+1}, \ldots, y_{3M})\tsr{B}^{-1}(z_{K+1}, \ldots, z_{2K},z_{2K+1}, \ldots, z_{3K}),
\]
where $\{y_{2M+1}, \ldots, y_{3M}\}\cup \{z_{2K+1}, \ldots, z_{3K}\} = \{a_1,\ldots, a_N\}$, $\{y_{2M+1}, \ldots, y_{3M}\}\cap \{z_{2K+1}, \ldots, z_{3K}\} = \emptyset$,
and $y_{i+2M} = a_{j}$ if $y_{i} = x_{j}$ for $\forall i\in\{1,\ldots, M\}$, $z_{i+2K} = a_{j}$ if $z_{i} = x_{j}$ for $\forall i\in\{1,\ldots, K\}$.

We denote the matrix multiplication of the matricized $\tsr{T}$ and $\tsr{C}$ as $\tsr{Z}$, and their relations can be shown as
\begin{align*}
\tsr{Z}(x_{1}, \ldots, x_{N}, a_1, \ldots, a_N) 
= 
\tsr{T}(x_1, \ldots, x_{2N}) \tsr{C}(x_{N+1}, \ldots, x_{2N}, a_1, \ldots, a_N) \\
= \tsr{A}(y_1, \ldots, y_{2M}) \tsr{B}(z_1, \ldots, z_{2K})
\tsr{A}^{-1}(y_{M+1}, \ldots, y_{2M},y_{2M+1}, \ldots, y_{3M}) 
\tsr{B}^{-1}(z_{K+1}, \ldots, z_{2K},z_{2K+1}, \ldots, z_{3K}) \\
= \tsr{I}_A(y_1, \ldots, y_M, y_{2M+1}, \ldots, y_{3M})
\tsr{I}_B(z_1, \ldots, z_K, z_{2K+1}, \ldots, z_{3K})
= \tsr{I}_{T}(x_{1}, \ldots, x_{N}, a_1, \ldots, a_N).
\end{align*}
Therefore, the theorem is proved.
\end{proof}

\vspace{4mm}

\subsection{Detailed Optimization Algorithms}
\label{subsec:appendix_opt}
In this Section, we provide detailed explanations on the union-find data structure for Einsum.
In addition, we provide detailed pseudo-codes for the distribution, Einsum fusion and dimension tree generation kernels discussed in Section~\ref{sec:opt}.

\textbf{Union-find for Einsum.}
The union-find (UF) representation for Einsum is a key ingredient of the optimization kernels. 
In the UF graph, each node represents one dimension of a specific tensor in the Einsum graph, and each 
edge represents a connection between two dimensions in an Einsum expression, where the connection is denoted by two dimensions sharing the same character in an Einsum string.
Downstream tasks can leverage this canonical representation of the Einsum graph for analysis. The algorithm of graph building is illustrated in Algorithm~\ref{alg:ufbuilder}. We use $\text{Einsum}\_\text{subscript}$ to denote the Einsum expression of an Einsum node, for example the string '$ij, jk \to ik$'. 

\begin{algorithm}[h]
\DontPrintSemicolon
\caption{BuildUF}
\label{alg:ufbuilder}
    \SetAlgoLined
    \SetKwInOut{Input}{input}\SetKwInOut{Output}{output}
    \Input{Einsum Graph: G}
    \Output{Union-find data structure: UF}

    Initialize a union-find data structure UF
    
    Initialize a map from Einsum character to tensor dimension DM
    
    \For{all einsum nodes $N$ in G}{
        \For {all characters C1 in N.einsum$\_$subscript} {
            \For {all characters C2 in N.einsum$\_$subscript} {
                \If{C1 == C2}{
                    UF.connect(DM[C1], DM[C2])
                }
            }
        }
    }

    return UF
\end{algorithm}

\begin{algorithm}[h]
\DontPrintSemicolon
\caption{Distribution}
\label{alg:distribute}
    \SetAlgoLined
    \SetKwInOut{Input}{input}\SetKwInOut{Output}{output}
    \Input{Graph: G}
    \Output{Distributed Graph: DG}
    \SetKwFunction{Distribute}{Distribute}
    DG = G

    \While{True} {
        \For{All DistributeOp nodes $\{$Ops$\}$ in DG} {
            \If{All Einsum nodes are topologically ahead of $\{$Ops$\}$} {
            return DG
            }
            \For{Op $\in$ $\{$Ops$\}$} {
                DG = \Distribute(Op, DG)
                \Comment{Distribute does 
                Einsum((a+b),c) $\rightarrow$ Einsum(a,c) + Einsum(b,c), where $+$ is the Op}
            } 
        }
    }
\end{algorithm}

\begin{algorithm}[h]
\DontPrintSemicolon
\caption{Einsum fusion}
\label{alg:fusion}
    \SetAlgoLined
    \SetKwInOut{Input}{input}\SetKwInOut{Output}{output}
    \Input{Einsum Tree: T}
    \Output{Fused Einsum Node: FN}
    \SetKwFunction{Linearize}{Linearize}
    \SetKwFunction{Declone}{Declone}

    LT = \Linearize(T)
    
    UF = BuildUF(LT)
    
    UF.Assign() \Comment{Assign each disjoint subset an unique character.}
    
    Init FN (sink: T.root, source: T.leaves)
    
    FN.genereateSubscript()
    \Comment{Generate FT.subscript based on input nodes' assigned characters.}
    
    FN = \Declone(FN)
    
    return FN
\end{algorithm}

\begin{algorithm}[h]
\DontPrintSemicolon
\caption{Dimension tree construction}
\label{alg:dt}
    \SetAlgoLined
    \SetKwInOut{Input}{input}\SetKwInOut{Output}{output}
    \Input{Einsum node List: NL, Input node list: IL}
    \Output{Updated Einsum node List: UL}
    \SetKwFunction{optconstraint}{Opt\_contraction\_path\_w\_constraint}

    $n=$ length(NL)

    UL = NL

    \For{i $\in$ $\{1,\ldots,n\}$} {
        contract\_order\_list = $[ \text{I[N],} \ldots \text{, I[i+1], I[1],} \ldots \text{, I[i-1]} ]$

        contract\_order\_list = part of contract\_order\_list where all elements are in UL[i].inputs

        UL[i] = \optconstraint(UL[i], contract\_order\_list)
    }
    return UL
\end{algorithm}

\vspace{4mm}

\subsection{Additional Benchmark Results}
We present the additional benchmark results for the kernels in the alternating minimization problems in Figure~\ref{fig:als_appendix}. For both CP and Tucker decompositions, we fix the tensor size in each mode and the decomposition rank, and compare the performance between AutoHOOT and other libraries with different input tensor order. As can be seen, the expressions generated with the dimension tree algorithm outperform all the other implementations.
For DMRG, we fix the physical dimension sizes and the MPO/MPS ranks, and compare the performance with different number of sites. As can be seen, AutoHOOT and Quimb have comparable performance.

\begin{figure*}[h]
\centering
\subfloat[NumPy, CPD, $s=R=30$]{
\includegraphics[width=.32\textwidth]{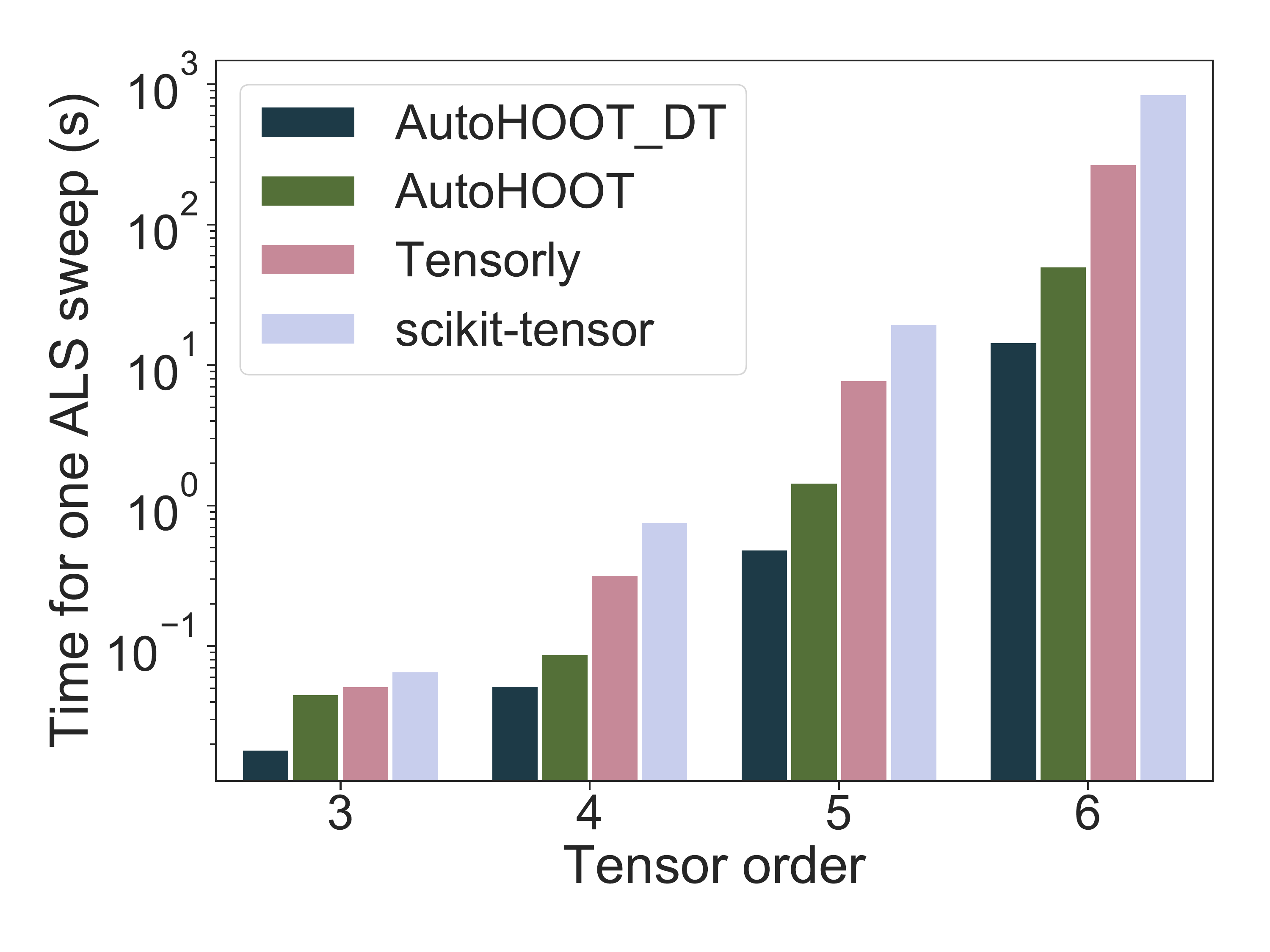}
\label{}
}
\subfloat[NumPy, Tucker, $s=30, R=10$]{
\includegraphics[width=.32\textwidth]{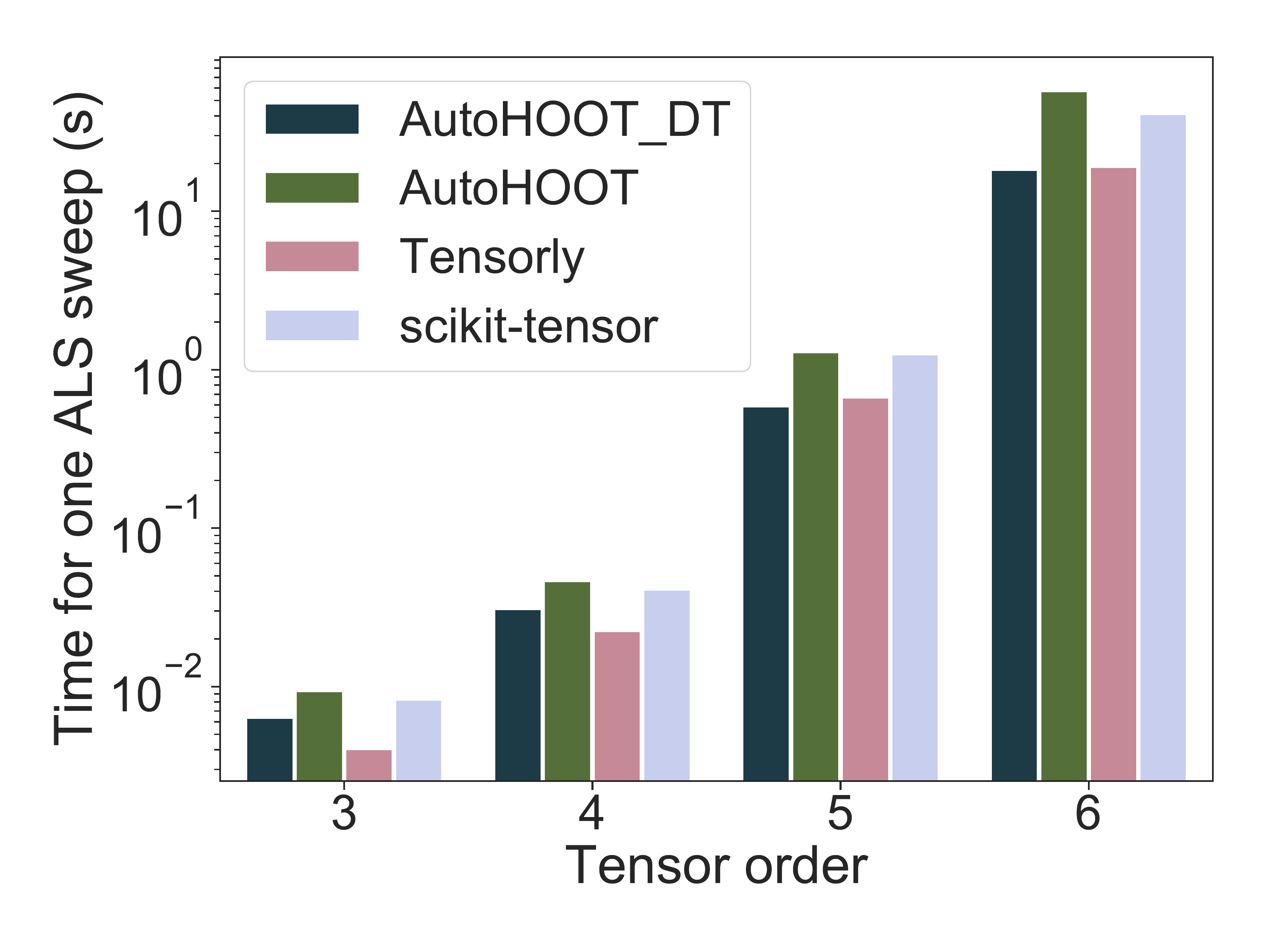}
\label{}
}
\subfloat[NumPy, DMRG, $s=R=40$]{
\includegraphics[width=.32\textwidth]{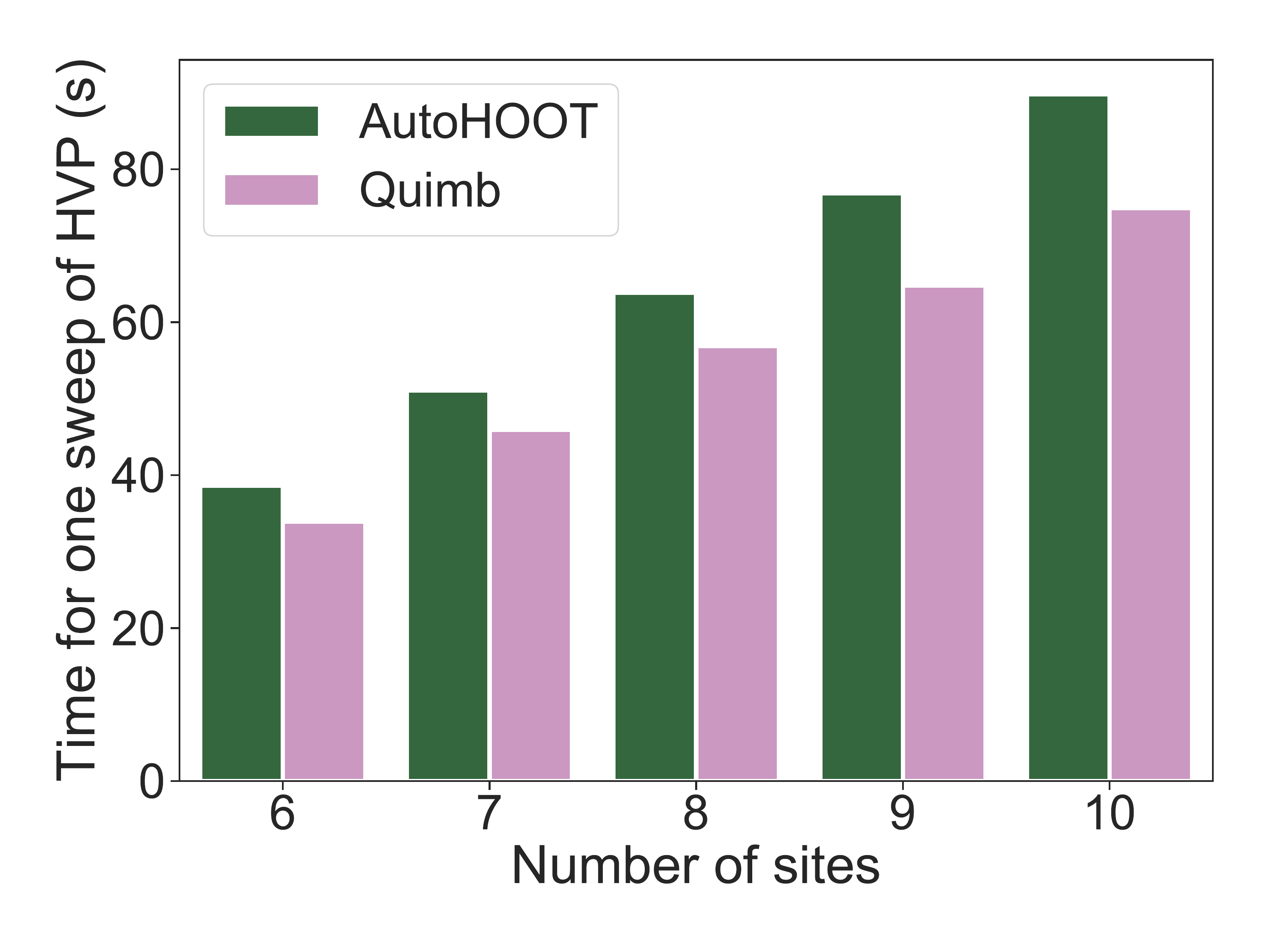}
\label{}
}

\centering
\subfloat[TensorFlow, CPD, $s=R=30$]{
\includegraphics[width=.32\textwidth]{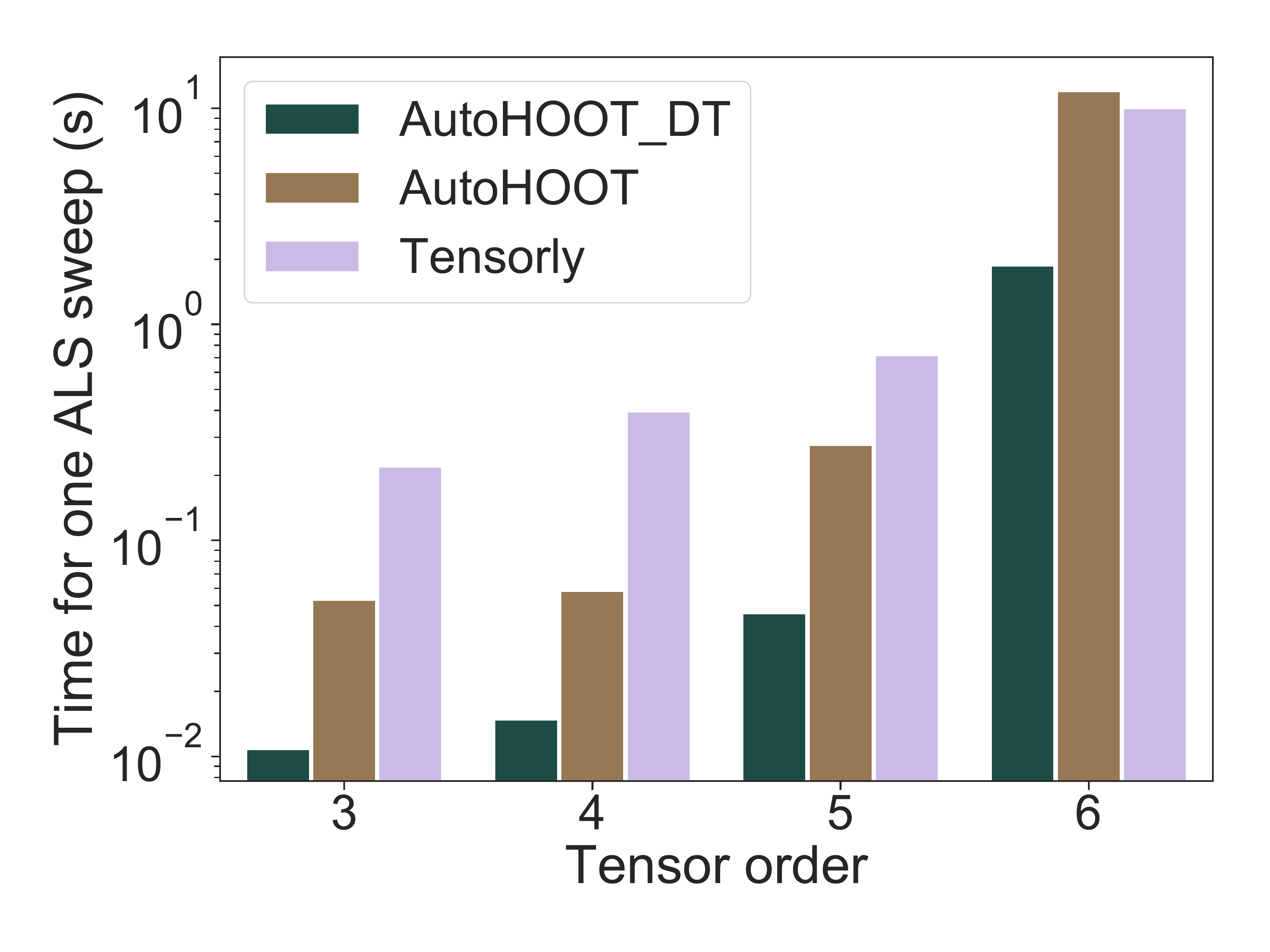}
\label{}
}
\subfloat[TensorFlow, Tucker, $s=30, R=10$]{
\includegraphics[width=.32\textwidth]{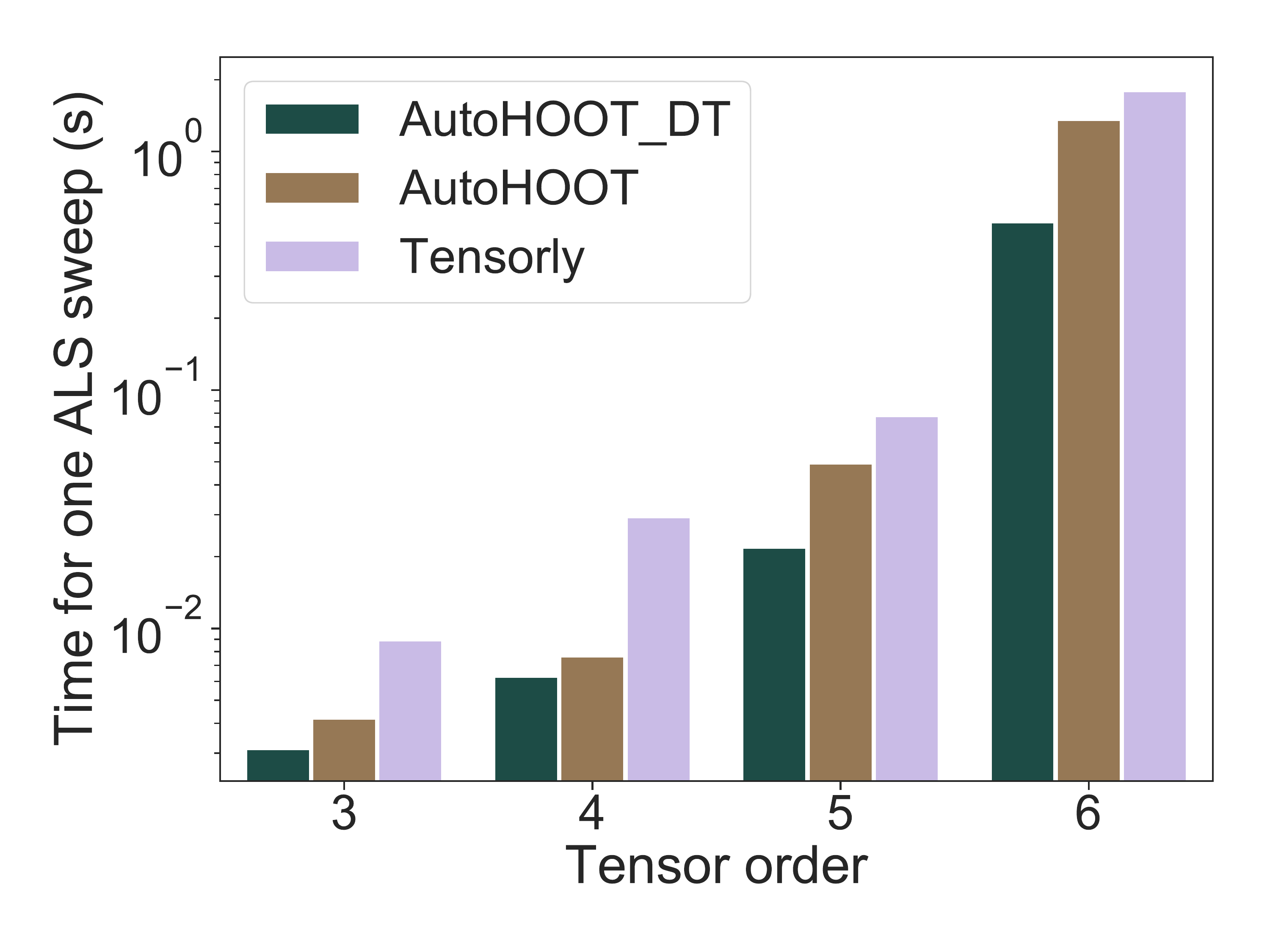}
\label{}
}
\subfloat[TensorFlow, DMRG, $s=R=40$]{
\includegraphics[width=.32\textwidth]{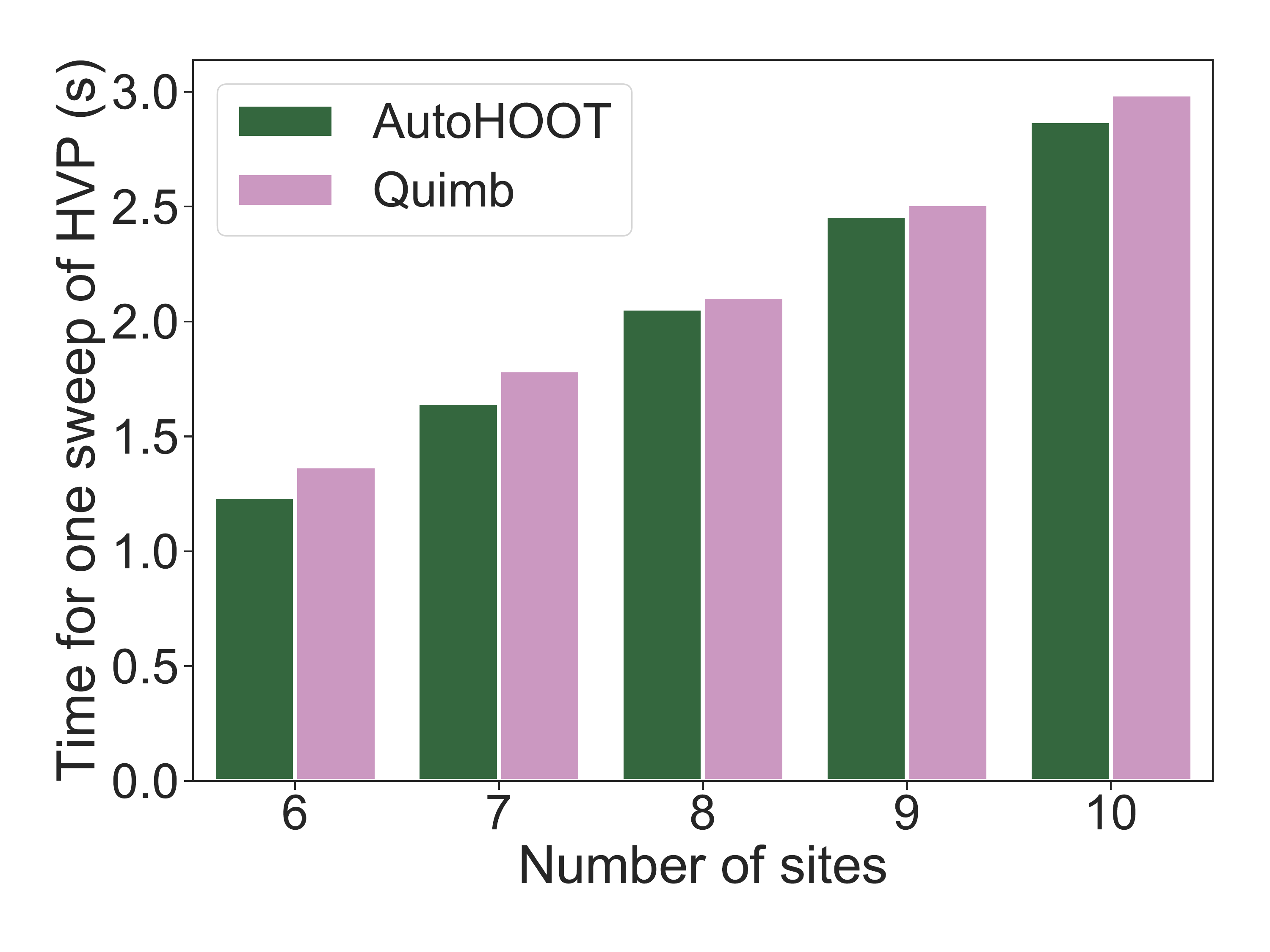}
\label{}
}
\caption{Additional AutoHOOT performance results for kernels in the alternating minimization. Experiments with NumPy backend are executed on the single process on CPU, and experiments with TensorFlow backend are executed on GPU.}
\label{fig:als_appendix}
\end{figure*}

\end{document}